\documentclass[11pt,a4paper,reqno,oneside]{amsproc}
\usepackage[left=1in,right=1in,bottom=1in,top=1in]{geometry}
\usepackage[utf8]{inputenc}
\usepackage[T1]{fontenc}
\usepackage{tikz,amsthm,amsmath,amstext,amssymb,amscd,epsfig,euscript, mathrsfs, dsfont,pspicture,multicol, mathtools,graphpap,graphics,graphicx,times,enumerate,subfig,sidecap,wrapfig,color}
\usepackage{hyperref}
\usepackage{bm}
\usepackage{dsfont}
\usepackage{pgfplots}
\usepackage[normalem]{ulem}
\usetikzlibrary{patterns}
\usepackage[normalem]{ulem}

\DeclareMathOperator{\dom}{dom}
\DeclareMathOperator{\ran}{ran}
\DeclareMathOperator{\spec}{spec}
\DeclareMathOperator{\supp}{supp}
\DeclareMathOperator{\pv}{p.v.}
\DeclareMathOperator{\Op}{Op}
\newcommand{\specd}{\spec_\mathrm{disc}}
\newcommand{\spece}{\spec_\mathrm{ess}}

\newcommand{\ddd}{\,\mathrm{d}}
\newcommand{\capa}{{\textrm{Cap}}}
\newcommand{\gM}{{\mathfrak{M}}}
\newcommand{\cS}{\mathcal{S}}

\def \rr {{\mathbb R}}
\def \cc {{\mathbb C}}
\def \nn {{\mathbb N}}

\def \zz {{\mathbb Z}}

\DeclareGraphicsExtensions{.png}
\numberwithin{equation}{section}

\def\var{\varphi}
\def\dom{\mathrm{dom}}

\def \cc {{\mathbb{C}}}

\newtheorem{definition}{Definition}[section]
\newtheorem{theorem}[definition]{Theorem}
\newtheorem{proposition}[definition]{Proposition}
\newtheorem{lemma}[definition]{Lemma}
\newtheorem{corollary}[definition]{Corollary}
\newtheorem{remark}[definition]{Remark}

\makeatletter

\title[The Landau--Dirac operator with  shell interactions]{The Landau--Dirac operator with  shell interactions: \\self-adjointness and clustering}

\makeatletter
\@namedef{subjclassname@2020}{%
  \textup{2020} Mathematics Subject Classification}
\makeatother
    
\author[B. Benhellal]{Badreddine Benhellal}
\address{(B. Benhellal) National Higher School of Mathematics, Scientific and Technology Hub of Sidi Abdellah, P. O. Box 75,  16093 Algiers, Algeria}
\email{badreddine.benhellal@nhsm.edu.dz}
\author[B. Vincent]{Vincent Bruneau}
\address{ (V. Bruneau) Institut de Math\'ematiques de Bordeaux, UMR 5251, Universit\' e de Bordeaux,  33405 Talence Cedex,  France.}
\email{vbruneau@math.u-bordeaux.fr}
\author[P. Miranda]{Pablo Miranda}
\address{(P. Miranda) Departamento de Matem\'{a}tica y Ciencia de la Computaci\'{o}n
Universidad de Santiago de Chile, Las Sophoras 173, Estaci\'{o}n Central, Santiago, Chile}
\email{pablo.miranda.r@usach.cl}
\date{\today}
\subjclass[2020]{Primary: 35Q40; Secondary 35P05, 45N50, 81Q10}
\keywords{magnetic Dirac operator, delta interaction, eigenvalue asymptotics}
\begin{document}
	\begin{abstract} 
	We consider the two-dimensional Dirac operator with  constant magnetic field that is  perturbed by a combination of electrostatic and Lorentz-scalar delta interactions with variable coefficients supported on a smooth closed curve. Self-adjointness is studied in the so-called non-critical and critical cases. In the non-critical case the essential spectrum is unchanged—it remains to be the set of the Landau--Dirac levels, the eigenvalues of infinite multiplicity of the unperturbed operator—while in the critical case an additional interval of essential spectrum emerges in the spectral gap containing zero. Our main result concerns the discrete spectrum in the non-critical case: using the pseudodifferential properties of the involved boundary integral operators, we show that the eigenvalues accumulate at each Landau--Dirac level at a rate governed by the logarithmic capacity of the curve.  A novel and surprising phenomenon is the change in the side of the accumulation depending on the position relative to the critical value. As a byproduct,  clusters of eigenvalues for a family of exterior boundary value problems are obtained via {\it confining} couplings; the infinite-mass boundary condition arises as a special case.  
	\end{abstract}
	\maketitle	
	\section{Introduction}
	Since the seminal works \cite{RaWa02,MeRo07}, the spectral properties of compactly supported perturbations of magnetic Hamiltonians have attracted considerable interest and have been extensively studied over the past two decades. In the classical setting of constant full-rank magnetic field --typically in two dimensions-- such perturbations are known to generate clusters of discrete eigenvalues accumulating exponentially fast  near the Landau levels (i.e., eigenvalues of infinite multiplicity corresponding to the free magnetic Hamiltonian). The asymptotic distribution of these eigenvalues and the rate of the accumulation are, notably, characterized in terms of the logarithmic capacity of the perturbation's support, as established in \cite{FiPu06}.
	
	In the context of magnetic  Schr\"odinger or Pauli operators, the side from which discrete  eigenvalues accumulate  to the Landau levels is closely related to the sign of the perturbation. This phenomenon appears uniformly across different types of compact  perturbations:  electromagnetic potentials \cite{ RaWa02, FiPu06, MeRo07, RoTa08, PuRo11}, obstacles  \cite{PushRoz07, Pe09, GoKaPe16}, and  delta interactions supported on a curve \cite{BHOP20}. In particular, the spectral behaviour induced by a delta interactions or obstacles is qualitatively analogous to that caused by an approximating electric potential.

For the Dirac operator, the picture is more intricate. This stems partly from the richer matrix-valued structure of the perturbations and partly from the subtleties involved in regularizing delta interactions by approximating potentials, which require (nonlinear) normalization of coefficients to ensure resolvent convergence;  see, e.g., \cite{CaLoMaTu23, BHT23} for the non-magnetic case.  In the special case of scalar perturbations of the Landau--Dirac operator in $\rr^2$, it is known that the side of the accumulation at the essential spectrum depends on the sign of the perturbation \cite{MeRo07}.  

Recently, in \cite{BM}, it was proved that for matrix-valued potentials of arbitrary sign, a particular matrix component dominates the spectral behaviour. For instance,  Lorentz-type potentials of the form $\sigma_3$ on a compact set,  induce eigenvalue accumulation exclusively from above at each Landau--Dirac level, despite $\sigma_3$ having eigenvalues $\pm 1$. 

The primary goal of this paper is to investigate whether this dominance phenomenon persists for perturbations of the Landau--Dirac operator by electrostatic and Lorentz-type delta interactions with variable coefficients    supported on smooth curves. In particular, a special case of our study will yield results for the obstacle problem with boundary conditions including the so-called {\it infinite mass} boundary condition, confirming a conjecture made in \cite{BM}. 

More precisely, for a constant magnetic field $b>0$ and  an  associated magnetic potential on $\rr^2$: 
\[
A(x)=\begin{bmatrix}
A_1(x) \\                                              
A_2(x) \\                                            
\end{bmatrix}
=\frac{b}{2} \begin{bmatrix}
-x_2 \\                                              
x_1 \\                                            
\end{bmatrix},
\]
where $(x_1,x_2)$ are the standard Cartesian coordinates of $x\in\rr^2$, we consider the Landau--Dirac operator with mass $m\geq 0$ defined in $L^2(\mathbb{R}^2 ; \mathbb{C}^2)$  by
\begin{align*} 
D_0 f &= D f, \quad
\dom\, D_0 =  \{f\in L^2(\rr^2; \cc^2):\ D f \in L^2(\rr^2; \cc^2)\}.
\end{align*}  
Here $D$ is the action given by 
\begin{align*}\label{LandauDirac}
    D:= -i\sigma\cdot \nabla_A +m\sigma_3, \quad  \nabla_A:= \nabla-iA,
\end{align*}
and $(\sigma_j)_{j=1,2,3}$ are the Pauli matrices defined by
\begin{equation*} 
 \sigma_1 =  \begin{pmatrix}
0 & 1\\                                              
1 & 0 \\                                            
\end{pmatrix}, \quad  \sigma_2 = \begin{pmatrix}
0 & -i\\                                              
i & 0 \\                                            
\end{pmatrix}, \quad  \sigma_3 =  \begin{pmatrix}
1 & 0\\                                              
0 & -1 \\                                            
\end{pmatrix}.
\end{equation*}
They satisfy the anticommutation relation
\[
\sigma_j\sigma_k +\sigma_k\sigma_j=2 \delta_{jk}I_2, \quad I_2= \begin{pmatrix}
1 & 0\\                                              
0 & 1 \\                                            
\end{pmatrix},
\]
and for $\sigma=(\sigma_1,\sigma_2)$ we use the standard notation
\[
\sigma\cdot x=\sigma_1 x_1 +\sigma_2 x_2.
\]
It is well known \cite[Theorem 7.2]{Thaller92} that $D_0$ is self-adjoint and its spectrum $\spec(D_0)$ consists of eigenvalues of infinite multiplicities, called \textit{Landau--Dirac levels},
\begin{align*}
    \mu_n= \begin{cases}
    \sqrt{2bn+m^2}, \,\,\,\, n\in\{0,1,2,\dots\},\\
    -\sqrt{2b|n|+m^2}, \,\, \,\,  n\in\{-1, -2,\dots\}.
\end{cases}
\end{align*}
We consider perturbations of $D_0$ by $\delta$-interactions supported on a smooth  closed curve $\Sigma\subset\rr^2$,  formally expressed as 
	\[
	D_{\epsilon, \tau}= D+ (\epsilon I_2 + \tau \sigma_3)\delta_\Sigma, 
	\]
 where  $\epsilon,\tau: \Sigma\to \rr$ are smooth coefficient functions satisfying $|\epsilon| \neq |\tau|$ everywhere on $\Sigma$.   As in the non-magnetic case, self-adjoint realizations of  $D_{\epsilon, \tau}$ in $L^2(\rr^2;\cc^2)$ are defined via appropriate transmission conditions on $\Sigma$, involving traces from either side of the curve $\Sigma$ (see Section \ref{secLDD} for more details). Due to the smoothness of the magnetic potential,  one expects the emergence of a natural classification:   the so-called non-critical case $\epsilon^2-\tau^2\neq 4$, and the critical case $\epsilon^2-\tau^2 = 4$, reflecting a fundamental dichotomy in the operator's behaviour.  It turns out that this intuition is indeed correct.  However, what is less expected is that the critical case also represents a turning point in determining the side on which eigenvalues accumulate.
Our main results for non-critical regime can be succinctly stated as follows:
 \begin{theorem}\label{main_Th0}
Let $\epsilon,\tau\in C^1(\Sigma)$ such that $(\epsilon^2-\tau^2)(\Sigma) \cap \{0,4\} = \emptyset$. Then $D_{\epsilon,\tau} $ is self-adjoint with domain embedded in the Sobolev space $ H^1(\rr^2\setminus\Sigma;\cc^2)$, and its essential spectrum remains unchanged: 
\[\spece (D_{\epsilon,\tau})=\spece (D_0)=\spec (D_0) = \{ \mu_n; \, n \in  \zz \}.\]
Furthermore, if the function
\[
\frac{\epsilon + \tau}{4 - (\epsilon^2-\tau^2)}: \,\Sigma\rightarrow \rr
\]
is strictly positive (resp. strictly negative), then near each Landau--Dirac level $\mu_n$, $ n \in \zz$, the operator $D_{\epsilon,\tau}$ has infinitely many eigenvalues accumulating to $\mu_n$ from above (resp. below), while only finitely many eigenvalues accumulate on the opposite side. 	
\end{theorem}
\begin{remark}\label{rqconf} As indicated in Section \ref{secLDD}, in the particularly confining  case $(\epsilon, \tau)=(0, 2)$ (resp. $(\epsilon, \tau)=(0, - 2)$) the transmission condition reduce to the so-called \textit{infinite mass} (resp. \textit{anti-infinite mass}) boundary condition.  Theorem \ref{main_Th0} confirms the conjecture in \cite{BM} that the Landau--Dirac operator  restricted to  the exterior of an obstacle with \textit{infinite mass} boundary condition exhibits an infinitely many discrete eigenvalues accumulating above (and finitely many below) each Landau-Dirac level $\mu_n$, $ n \in \zz$, with the eigenvalue distribution governed by \eqref{Asymp}. The \textit{anti-infinite mass} boundary condition inverts the side of accumulations, see Corollary \ref{ClusConf}.
\end{remark}
\begin{figure}
        \centering
\begin{tikzpicture}[scale=0.9]

\draw[->] (-3.2,0) -- (3.2,0) node[right] {$\epsilon$};
\draw[->] (0,-3.2) -- (0,3.2) node[above] {$\tau$};

\fill[gray!25, opacity=0.7, domain=-2:2,smooth,variable=\x]
  (-3,3)
  --(2,3)
  --(2,-2)
   -- cycle;
   
\fill[gray!25, opacity=0.7, domain=2:3,smooth,variable=\x]   
(2,3)
--plot ({\x},{sqrt(\x*\x-4)})
--(3,3)
 -- cycle;
 
 \fill[gray!25, opacity=0.7, domain=2:3,smooth,variable=\x]   
(2,-2)
--plot ({\x},-{sqrt(\x*\x-4)})
--(3,-3)
 -- cycle;


\fill[gray!25, opacity=0.7, domain=-3:-2,smooth,variable=\x]
 plot ({\x},{-sqrt(\x*\x-4)})
 --(-3,0)
  -- plot ({\x},{sqrt(\x*\x-4)})
   --(-3,0)
  -- cycle;

\fill[pattern=north east lines,domain=2:3,smooth,variable=\x]
  plot ({\x},{-sqrt(\x*\x-4)})
  --(3,0)
  -- plot ({\x},{sqrt(\x*\x-4)})
  --(3,0)
  -- cycle;
  
\fill[pattern=north east lines,domain=-2:3,smooth,variable=\x]
  (-2,-3)
  --(-2,2)
  --(3,-3)
   -- cycle;
   
\fill[pattern=north east lines,domain=-3:-2,smooth,variable=\x]
(-3,-3)
--plot ({\x},-{sqrt(\x*\x-4)})
--(-2,-3)
 -- cycle;

\fill[pattern=north east lines,domain=-3:-2,smooth,variable=\x]
(-3,3)
--plot ({\x},{sqrt(\x*\x-4)})
--(-2,2)
 -- cycle;

\draw[red, thick] (-3,3) -- (3,-3); 

\draw[red, thick] (-3,-3) -- (3,3); 

\draw[red, thick,domain=-3:-2,smooth] plot (\x,{sqrt(\x*\x-4)});
\draw[red, thick,domain=2:3,smooth] plot (\x,{sqrt(\x*\x-4)});
\draw[red,thick,domain=-3:-2,smooth] plot (\x,{-sqrt(\x*\x-4)});
\draw[red,thick,domain=2:3,smooth] plot (\x,{-sqrt(\x*\x-4)});

\node at (1.5,2.3) {$>0$};
\node at (-2.5,0.3) {$>0$};
\node at (-1.5,-2.3) {$<0$};
\node at (2.5,0.3) {$<0$};

\end{tikzpicture}
\caption{When the image of $(\epsilon, \tau)$ is outside the red lines (critical cases), eigenvalues of $D_{\epsilon, \tau}$ accumulate at the Landau--Dirac levels from above (grey region) or from below (hatched region).}
        \label{fig1}
\end{figure}
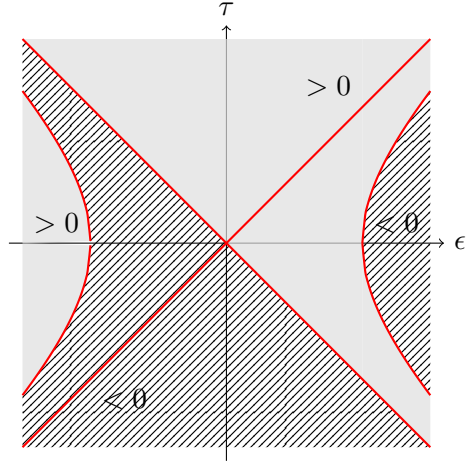

In the previous theorem, the distribution of these eigenvalues follows an asymptotic behaviour similar to classical results, involving the logarithmic capacity of the support of the perturbation (here $\Sigma$), see formula \eqref{Asymp} for precise statements.  Remarkably, unlike the scenario for compactly supported scalar perturbations of magnetic Hamiltonians  \cite{FiPu06, MeRo07, GoKaPe16, BrRa20} (including the Landau Hamiltonian with delta interactions \cite{BEHL20}), where the side of accumulation is determined solely by the sign of the potential, the Dirac operator $D_{\epsilon, \tau}$ exhibits a notably richer behaviour as can be seen in Figure \ref{fig1}.  In fact, this behaviour also differs from the one observed for a measurable compactly supported perturbation 
\[
V= \begin{pmatrix}
V_1 & 0\\                                              
0 & V_2 \\                                            
\end{pmatrix}
\]
 of $D_0$, where the side of the accumulation is determined by the sign of $V_1$ (if $\text{supp} V_2 \subset \text{supp} V_1$), see \cite{BM}.  In our situation, we formally have $V_1= (\epsilon  + \tau )\delta_\Sigma$ and $V_2= (\epsilon - \tau )\delta_\Sigma$. However, even if both $(\epsilon + \tau )$ and $(\epsilon - \tau )$ are positive on $\Sigma$, accumulation can occur on the 'opposite' side depending on their combined magnitude (it suffices that their product is larger than $4$), a phenomenon unique to the Dirac operator $D_{\epsilon, \tau}$. This insight naturally leads to a detailed study of the critical case $\epsilon^2-\tau^2=4$. As usual, the first task is to understand the self-adjointness and basic spectral properties, but even at this stage our analysis reveals intriguing spectral features:

\begin{theorem} Suppose that $\epsilon^2(s)-\tau^2(s)=4$  for each $s\in \Sigma$. Then,  $D_{\epsilon,\tau}$ is self-adjoint, but its domain does not embed into  $H^1(\rr^2\setminus\Sigma;\cc^2)$. Moreover, its essential spectrum enlarges to  $\{ \mu_n; \, n \in  \zz \}\cup  I_\Sigma$, where 
\begin{align*}
I_\Sigma:=  \ran\left(-m\frac{\tau}{\epsilon}\right).
\end{align*}
\end{theorem}

The appearance of the additional interval $I_\Sigma$ in the essential spectrum is itself noteworthy:  it shows that, at the critical threshold $\epsilon^2-\tau^2=4$, the delta interaction is strong enough to generate essential spectrum of its own, partially filling the gap $(-m,m)$ whenever $m>0$ and $\frac{\tau}{\epsilon}$ is non-constant along $\Sigma$; when $\frac{\tau}{\epsilon}$ is constant this additional contribution collapses to the single point $\{-m\frac{\tau}{\epsilon}\}$.  The Landau--Dirac levels alone are recovered when $m=0$.  This also generalizes the main result of \cite{BHOP20}, where the appearance of the extra point $\{-m\frac{\tau}{\epsilon}\}$ in the essential spectrum was observed in the absence of the magnetic field and for constant interaction coefficients.

Notice that even on bounded domains the emergence of essential spectrum on an interval is possible.  Indeed, for a  smooth bounded domain $\Omega\subset\rr^2$ such that $\Sigma\subset\Omega$, consider the restriction of  the Dirac operator $D_{\epsilon,\tau}$  to $L^2(\Omega;\cc^2)$  subject to a  suitable local  boundary condition on $\partial \Omega$ (e.g., the infinite mass boundary condition).  Then, by  adapting \cite{BBHP} we should get that the essential spectrum of this operator coincides with $I_\Sigma$. This would give us an example of a self-adjoint operator on bounded domains of $\rr^2$  with  non-trivial essential spectrum, extending the result of  \cite{BBHP} to a two dimensional setting.

In three dimensions, a related phenomenon arise when $\Sigma\subset\rr^3$ is a compact surface. In this situation an additional interval of essential spectrum  appears \cite{BP24}, namely, 
\begin{equation*}
 J_\Sigma= \left[-\frac{\tau}{\epsilon}m -\frac{ A_\Sigma}{2|\epsilon|},-\frac{\tau}{\epsilon}m   +\dfrac{ A_\Sigma}{2|\epsilon|}\right],
\end{equation*}
where 
\[
A_\Sigma:=\max_{x\in\Sigma}\big|\kappa_1(x)-\kappa_2(x)\big|
\]
involves the difference of the principal curvatures  at $x\in\Sigma$. For the analogue of $D_0$ (i.e., the Magnetic Dirac Hamiltonian in 3D), the essential spectrum consists of purely absolutely continuous bands with thresholds given by $\mu_n$,  that is, $\spec D_0= (-\infty,-\sqrt{m^2+2b}]\cup [m,\infty)$.  If we consider constant critical interactions coefficients, by  mimicking our analysis carried below together with the one from \cite{BP24}, we expect to arrive at the conclusion that the essential spectrum of 3D the analogue of  $D_{\epsilon,\tau}$ is given by the union $J_\Sigma \cup \spec D_0 $  with $J_\Sigma$ as above  (the case of variable interactions is more subtle as the interval $[-A_\Sigma, A_\Sigma]$ is the essential spectrum of a $2\times2$ block pseudodifferential operator on $\Sigma$, see \cite[pages 250-251]{BP24}).

 Concerning the accumulation of discrete eigenvalues in the critical case, we notice that the emergence of this additional essential spectrum influences the rest of the spectrum. For example, in the non-magnetic case when $\Sigma$ is a circle and $\epsilon$, $\tau$ are constants, it was shown in  \cite{BCF} that  $\{-m\frac{\tau}{\epsilon}\}$ serves as an accumulation point for a double sequence of eigenvalues within $(-m,m)$, approaching from both the left and the right, while no eigenvalues accumulate at $\pm m$. In the magnetic case, we may therefore expect the possibility of spectral accumulation on both sides of $I_\Sigma$, and the distribution of the eigenvalues near the Landau--Dirac levels could be affected by this.

Our approach to addressing the clustering of eigenvalues consists on showing that the eigenvalue counting function near each Landau level is controlled, via min-max principle and operator inequalities,  by a Toeplitz-type effective Hamiltonian involving the resolvents difference $\Upsilon:= (D_{\epsilon,\tau}-\zeta)^{-1}-(D_{0}-\zeta)^{-1}$ for some fixed $\zeta\in \rr\cap\rho(D_{\epsilon,\tau})\cap\rho(D_0)$ (see Section~\ref{SEC51}). This reduction relies crucially on the fact that, in the non-critical case, the effective Hamiltonian is compact thanks to the compactness of $\Upsilon$. In the critical case, however,  the effective Hamiltonian is no longer compact as it involves an unbounded operator on $\Sigma$ (see Lemma \ref{L BS}).   This prevents us from obtaining the analogue of Lemma \ref{le1}, where the $O(1)$  bound crucially depends on compactness. One might expect an effective Hamiltonian of a different nature that accounts for an interaction between the Landau--Dirac levels and $I_\Sigma$. We do not address this  issue in the current paper.

\subsection*{Organization of the paper} 
In Section~\ref{S2}, we introduce the functional spaces and boundary integral operators adapted to the study of  delta interactions for magnetic Hamiltonians in $\rr^2$. Using the expression of the Green function of $D_0$ in terms of Kummer confluent hypergeometric functions, we establish the pseudodifferential properties of boundary integral operators needed in our context. The operator $D_{\epsilon, \tau}$ is then defined in Section~\ref{secLDD} for $|\epsilon | \neq | \tau |$, and in the confining case ($\epsilon^2-\tau^2 = -4$) its description as decoupled boundary problems is given  (see Lemma~\ref{TC_alter}). 
We study the self-adjointness and state Krein-type resolvent formulas in Section~\ref{S4} for the non-critical case (see Theorem \ref{main_th1})  and in Section~\ref{S6} for the critical case (see Theorem \ref{main_th2} and Proposition \ref{BS and Krein 2}). These Sections also discuss the essential spectrum, with particular attention to the critical case (see Theorem \ref{Ess Spec}). The distribution of the discrete spectrum near Landau--Dirac levels in the non-critical case is covered in Section~\ref{S5} (see Theorem \ref{Cluster}). This analysis combines the earlier results of \cite{PushRoz07,BM} with those of Section~\ref{S4}.

 \section{The Landau--Dirac operator and associated boundary integral operators}\label{S2}
Let $\rho(D_0)$ be the resolvent set of $D_0$. Then for $z\in\mathbb C$ one has 
\begin{align*}
D_0^2 - z^2 \;\;
&=\; \begin{pmatrix} -\nabla_A^2 - b + m^2 - z^2 & 0 \\ 0 & -\nabla_A^2 + b + m^2 - z^2\end{pmatrix},
\end{align*}
where $-\nabla_A^2$ is the Landau Hamiltonian. Recall that (cf. \cite{AvHeSi78})
\[
\spec ( -\nabla_A^2)=\spece ( -\nabla_A^2)= \{b(2n+1): n\in \{0,1,2,\dots\} \}.
\]
For any $w\in\rho(-\nabla_A^2)$ the integral kernel of $(-\nabla_A^2-w)^{-1}$  on $L^2(\rr^2;\cc^2)$ is given by (see \cite{DoMaMa75,HoSmi2002})
\begin{equation*}
\psi_w(x,y) \;:=\; \frac{1}{4\pi}\,\Gamma\!\Big(\tfrac{b-w}{2b}\Big)\, e^{-\frac{b}{4}|x-y|^2 - i\frac{b}{2} x\wedge y}\, \mathbb{U}\!\Big(\tfrac{b-w}{2b},\,1,\,\tfrac{b}{2}|x-y|^2\Big),
\end{equation*}
where $\mathbb{U}(\alpha,\beta,\cdot)$ is the Kummer confluent hypergeometric function  \cite[Section (13.2)]{Olver}, $\Gamma$ is  the Gamma function  and $x\wedge y = x_1y_2 - x_2y_1$. 
Define 
\begin{equation}\label{pi}
\pi_+ :=\tfrac{I_2+\sigma_3}{2}=
\begin{pmatrix}
1 & 0 \\
0 & 0
\end{pmatrix},
\qquad
\pi_- :=\tfrac{I_2-\sigma_3}{2}=
\begin{pmatrix}
0 & 0 \\
0 & 1
\end{pmatrix},
\end{equation}
and 
\begin{equation*}
G_z(x,y) \;=\; \psi_{b-m^2+z^2}(x,y)\,\pi_+ \;+\; \psi_{-(b+m^2)+z^2}(x,y)\,\pi_-.
\end{equation*}
Then, for any $z\in\rho(D_0)\setminus{\{- m}\}$
\begin{align*}
    (D_0 -z)^{-1}g(x)&=(D_0+z)(D_0^2-z^2)^{-1}g(x)\\
    &=\int_{\rr^2}(D+z)G_z(x,y)g(y)\ddd y, \quad g\in L^2(\rr^2).\\
\end{align*}
Set
\begin{equation*}
\varphi_z(x,y):=((D+z)G_z)(x,y).
\end{equation*}
Notice that, since $\overline{\psi_z(x,y)}=\psi_{\Bar{z}}(y,x)$, using the anti-commutations properties of $\sigma_j$ one easily obtains
\begin{align}\label{var_ast}
     {\varphi_z(x,y)}^*=\varphi_{\Bar{z}}(y,x).
    \end{align}
To give a more explicit expression for $\varphi_z$ we will need the following notation: 
\begin{equation}\label{notation}
z_m\;:=\;\frac{m^2-z^2}{2b}, \qquad \zeta\;:=\;\tfrac{b}{2}|x-y|^2, \qquad F(x,y)\;:=\;e^{-\zeta/2 \,-\, i\frac{b}{2}x\wedge y}. 
\end{equation}

\begin{lemma}\label{prop:explicit_varphi}
For every $z\in\rho(D_0)\setminus\{- m\}$, the integral kernel of $(D_0-z)^{-1}$ is given by
\begin{equation}\label{varphi_z}
\begin{aligned}
\var_z(x,y)\;=\;\frac{F(x,y)}{4\pi}\bigg\{\;
&(m+z)\,\Gamma(z_m)\,\mathbb{U}(z_m,1,\zeta)\,\pi_+ \\
\;+\;&(z-m)\,\Gamma(z_m+1)\,\mathbb{U}(z_m+1,1,\zeta)\,\pi_- \\
\;+\;&i\,b\,\Gamma(z_m+1)\,\mathbb{U}(z_m+1,2,\zeta)\,\sigma\!\cdot\!(x-y)\bigg\}.
\end{aligned}
\end{equation}
\end{lemma}

\begin{proof}
Let us begin by introducing the  creation and annihilation operators 
\begin{equation*}
a := (-i \partial_{x_1} -A_1) + i (-i \partial_{x_2} -A_2) , \qquad a^* := (-i \partial_{x_1} -A_1) - i (-i \partial_{x_2} -A_2), 
\end{equation*}
so that
\[
D \;=\; \begin{pmatrix} m & a^* \\ a& -m \end{pmatrix}.
\]
Then the  diagonal entries of $(D+z)G_z$ are 
\[
\bigl[\varphi_z\bigr]_{11} = (m+z)\,\psi_{b-m^2+z^2}, \qquad
\bigl[\varphi_z\bigr]_{22} = (z-m)\,\psi_{-(b+m^2)+z^2},
\]
and  the off-diagonal terms are
 $a\bigl[F\,\mathbb{U}(z_m,1,\zeta)\bigr]$ and $a^*\bigl[F\,\mathbb{U}(z_m+1,1,\zeta)\bigr]$.

Introduce the complex coordinates $\omega:=x_1+ix_2$, $\omega':=y_1+iy_2$. Then $|x-y|^2=(\omega-\omega')(\overline{\omega-\omega'})$ and  $ix\wedge y = \tfrac{1}{2}(\overline\omega\omega'-\omega\overline\omega')$.  Moreover,  the differential operators $a$ and $a^*$ are connected to the complex derivatives 
$
\partial:= \frac12(\partial_{x_1} - i \partial_{x_2})$ and $ \overline{\partial}:= \frac12(\partial_{x_1} + i \partial_{x_2}),
$
by the relations: 
\begin{equation}\label{a_partial}
a =-i\left(2\bar\partial +\tfrac{b\omega}{2}\right), \qquad
a^*= -i(2\partial-\tfrac{b\overline\omega}{2}).
\end{equation}
A direct calculation gives
\begin{equation*}
\bar\partial F \;=\; -\tfrac{b\omega}{4}\,F, \qquad
a F \;=\; 0, \qquad
a^*F \;=\; i\,b(\overline\omega-\overline\omega')\,F.
\end{equation*}
Since $\bar\partial\zeta=\tfrac{b}{2}(\omega-\omega')$ and $\frac{d}{d\zeta}\mathbb{U}(\alpha,1,\zeta) = -\alpha\,\mathbb{U}(\alpha+1,2,\zeta)$, we get
\begin{equation*}
a\bigl[F\,\mathbb{U}(z_m,1,\zeta)\bigr] \;=\; i\,z_m\,b\,(\omega-\omega')\,F\,\mathbb{U}(z_m+1,2,\zeta).
\end{equation*}
Moreover
\[
a^*\bigl[F\,\mathbb{U}(z_m+1,1,\zeta)\bigr] \;=\; -i\,b(\overline\omega-\overline\omega')\,F\,\Bigl[\mathbb{U}(z_m+1,1,\zeta) \,+\, (z_m+1)\,\mathbb{U}(z_m+2,2,\zeta)\Bigr].
\]
Using the relation \cite[(13.3.9)]{Olver} we conclude  that 
\begin{equation*}
a^*\bigl[F\,\mathbb{U}(z_m+1,1,\zeta)\bigr] \;=\; -i\,b(\overline\omega-\overline\omega')\,F\,\mathbb{U}(z_m+1,2,\zeta).
\end{equation*}
Using $\Gamma(z+1)=z\,\Gamma(z)$ we finish the proof.
\end{proof}

\begin{lemma}\label{lemma:varphi_asymp}
For every $z\in\rho(D_0)\setminus\{-m\}$ we have the decomposition
\begin{equation*}
\var_z(x,y) \;=\; \theta(x-y) \;-\; \frac{m\,\sigma_3 + z\,I_2}{2\pi}\,\log|x-y| \;+\; \widetilde\theta_z(x,y),
\end{equation*}
where
\begin{equation*}
\theta(x) \;:=\;\frac{i}{2\pi} \,\sigma\!\cdot\!\frac{x}{|x|^2},
\end{equation*}
and  $\widetilde\theta_z:\rr^2\times\rr^2\to\cc^{2\times 2}$ is continuous in a neighbourhood of the diagonal $x=y$ and $C^\infty$ outside it.
\end{lemma}

\begin{proof}
First notice that from  \cite[(13.2.19), (13.5.17)]{Olver}, as $\zeta\to 0^+$,
\begin{equation*}
\mathbb{U}(\alpha,1,\zeta) \;=\; -\frac{1}{\Gamma(\alpha)}\Bigl[\log\zeta + \psi(\alpha) + 2\gamma\Bigr] \,+\, O(\zeta\log\zeta),
\end{equation*}
\begin{equation}\label{U_asymp_2}
\mathbb{U}(\alpha,2,\zeta) \;=\; \frac{1}{\Gamma(\alpha)\,\zeta} \,+\, O(\log\zeta),
\end{equation}
where $\gamma$ is the Euler constant and $\psi$ is the digamma function. Moreover, in \eqref{notation}  $F(x,y) = 1 + O(|x-y|)$. Then 
substituting \eqref{U_asymp_2} into the off-diagonal terms of \eqref{varphi_z} gives 
\begin{align*}
\frac{F}{4\pi}\,i\,b\,\Gamma(z_m+1)\,\mathbb{U}(z_m+1,2,\zeta) =& \frac{i\,b}{4\pi\,\zeta}+O(\log \zeta) \\
=& \frac{i}{2\pi\,|x-y|^2}+O(\log |x-y|),
\end{align*}
and multiplication by $\sigma\!\cdot\!(x-y)$ yields 
\[
\frac{i}{2\pi}\,\sigma\!\cdot\!\frac{x-y}{|x-y|^2}=\theta(x-y).
\]
For the diagonal terms, we have 
\begin{align*}
&\frac{F}{4\pi}\Bigl[(m+z)\,\Gamma(z_m)\,\mathbb{U}(z_m,1,\zeta)\,\pi_+ + (z-m)\,\Gamma(z_m+1)\,\mathbb{U}(z_m+1,1,\zeta)\,\pi_-\Bigr]
\\
=& -\frac{\log\zeta}{4\pi}\bigl[(m+z)\,\pi_+ + (m-z)\,\pi_-\bigr]+O(1).
\end{align*}
Since $(m+z)\pi_+ + (z-m)\pi_- = m\,\sigma_3 + z\,I_2$, this contribution equals
\[
-\frac{m\,\sigma_3 + z\,I_2}{2\pi}\,\log|x-y| \;+\; R(x,y),
\]
where the matrix-valued remainder satisfies $R(x,y)=O(1)$.
\end{proof}

\subsection{Functional spaces associated with $D_0$} 
Throughout the paper, $\Omega_+ \subset \mathbb{R}^2$ is a bounded domain with $C^\infty$-smooth boundary $\Sigma=\partial\Omega_+$, and denote its complement by   $\Omega_- := \mathbb R^2 \setminus \overline{\Omega_+}$. We write $\nu=(\nu_1,\nu_2)^\top$ for the unit normal vector field on $ \Sigma$  pointing outward from $\Omega_+$. For $s\in[0,1]$ we denote by $H^{s}(\Omega_\pm;\cc^2)$ the standard $L^2$-based Sobolev space of order $s$ (with the convention that $H^{0}(\Omega_\pm;\cc^2)=L^{2}(\Omega_\pm;\cc^2)$),  and  by $H^s_{\sigma, A}(\Omega_\pm; \cc^2)$ the spaces:
\begin{align*}
		H^s_{\sigma, A}(\Omega_\pm; \cc^2):=\Big\{f\in H^s(\Omega_\pm; \cc^2):\ \sigma\cdot\nabla_A f\in L^2(\Omega_\pm; \cc^2)
		\Big\},
	\end{align*}
	which becomes a Hilbert space if endowed with the norm
	\[
	\| f\|^2_{H^s_{\sigma, A}(\Omega;\cc^2)}:= \| f\|^2_{H^s(\Omega;\cc^2)}+ \|\sigma\cdot\nabla_A f\|^2_{L^{2}(\Omega;\cc^2)}.
	\]
 We also define the non-magnetic analogue of $H_{\sigma, A}(\Omega_\pm;\cc^2)\equiv H^0_{\sigma, A}(\Omega_\pm;\cc^2)$,
	\begin{align*}
		H_{\sigma}(\Omega_\pm; \cc^2):=\Big\{f\in L^2(\Omega_\pm; \cc^2):\ \sigma\cdot\nabla f\in L^2(\Omega_\pm; \cc^2)
		\Big\},
	\end{align*}
 endowed with the norm
	\[
	\| f\|^2_{H_{\sigma}(\Omega_\pm;\cc^2)}:= \| f\|^2_{L^2(\Omega_\pm;\cc^2)}+ \|\sigma\cdot\nabla f\|^2_{L^{2}(\Omega_\pm;\cc^2)}.
	\]
	As we will see in Section \ref{secLDD}, the domain of the Dirac operator $D_{\epsilon,\tau}$ is a subset of $H_{\sigma, A}(\Omega_+; \cc^2)\oplus H_{\sigma, A}(\Omega_-; \cc^2)$. Hence,  to have a well define transmission condition along $\Sigma$, we first need to extend  the trace operator to the magnetic Sobolev spaces $H_{\sigma, A}(\Omega_\pm;\cc^2)$. 

Let us start by noticing that from \cite[Theorem 1.24]{BLRS23} it can be seen that 
\[
\dom D_0=\{ f\in H^1(\mathbb R^2; \cc^2):\ (\sigma\cdot A ) f\in L^2(\mathbb R^2)\}.
\]
In particular, for any $\eta \in C_0^\infty(\mathbb R^2)$   we have that $\eta H^1(\mathbb R^2; \cc^2)\subset \dom D_0.$  Therefore  the Dirichlet trace 
\[
\gamma_D^\Sigma: \dom D_0\longrightarrow H^{\frac{1}{2}}(\Sigma;\cc^2), \quad  f\mapsto f|_{\Sigma}, 
\] 
is a bounded surjective linear operator.

We recall that the trace maps $\gamma_D^\pm: H^1(\Omega_\pm; \mathbb{C}^2) \rightarrow H^{\frac{1}{2}}(\Sigma; \mathbb{C}^2)$ admit unique continuous extensions from $H_\sigma(\Omega_\pm; \mathbb{C}^2)$ to $H^{-\frac{1}{2}}(\Sigma; \mathbb{C}^2)$; cf. \cite[Lemma 2.3]{BFSV}. Moreover, if $u\in H_{\sigma}(\Omega_\pm;\cc^2) $ with $\gamma^\pm_D u\in H^{\frac{1}{2}}(\Sigma;\cc^2)$ then $u\in H^1(\Omega_\pm;\cc^2)$; cf. \cite[Lemma 2.4]{BFSV}.  Consequently, one readily obtains the following result.
\begin{lemma}\label{tracemap} The Dirichlet trace operators $\gamma_D^\pm$ extend to continuous maps 
\[
\gamma_D^\pm:H_{\sigma, A}(\Omega_\pm; \mathbb{C}^2)\rightarrow H^{-\frac{1}{2}}(\Sigma; \mathbb{C}^2). 
\] 
Moreover, the following holds:
 If $u\in H_{\sigma, A}(\Omega_\pm;\cc^2) $ with $\gamma^\pm_D u\in H^{\frac{1}{2}}(\Sigma;\cc^2)$, then $u\in H^1(\Omega_\pm;\cc^2)$.
\end{lemma} 
\begin{proof}
    It is clear that $H_{\sigma, A}(\Omega_+;\cc^2)=H_{\sigma}(\Omega_+;\cc^2)$ and  the statement for $\gamma_D^+$ follows directly.   For $\gamma_D^-$ notice that  $H_{\sigma, A}(\Omega_-;\cc^2)$ coincides locally with $H_{\sigma}(\Omega_-;\cc^2)$ and in view of the elliptic regularity of $D_0$, away from the boundary, functions in $H_{\sigma, A}(\Omega_-;\cc^2)$ possess the $H^1$-regularity.  That is,   for any $R>0$ such that $\overline{\Omega_+}\subset B(0,R)$ and any smooth partition of unity $\psi_1,\psi_2:\overline{\Omega_-} \longrightarrow [0,1]$ with 
     \begin{align*}
    \psi_1(x)= \begin{cases}
    1, \,\,\,\,  x\in \overline{\Omega_-}\cap B(0,2R),\\
    0, \,\, \,\,  x\in \rr^2 \setminus B(0,3R),
\end{cases}
\end{align*}
 one has
    \begin{align*}
     \psi_1 u\in H_{\sigma}(\Omega_-;\cc^2) \quad\text{and}\quad \psi_2u\in H^1(\Omega_-;\cc^2),
     \end{align*}
     for any $u\in H_{\sigma, A}(\Omega_-;\cc^2)$.  Then using \cite[Lemma 2.3]{BFSV} we define:
\[
\gamma_D^-(u):=\gamma_D^-(\psi_1 u)+\gamma_D^-(\psi_2 u).
\]     
For the second assertion  of the lemma, since    $u\in H_{\sigma, A}(\Omega_-;\cc^2), $  we have that   $\psi_2u\in H^1(\Omega_-;\cc^2) $, then $\gamma_D^-(\psi_2 u)\in H^{\frac{1}{2}}(\Sigma;\cc^2)$.  Now, if in addition   $\gamma^-_D u\in H^{\frac{1}{2}}(\Sigma;\cc^2)$ we obtain that 
$\gamma_D^-(\psi_1 u)=\gamma_D^-(u)-\gamma_D^-(\psi_2 u)\in H^{\frac{1}{2}}(\Sigma;\cc^2)$. Since   $\psi_1 u\in H_{\sigma}(\Omega_-;\cc^2)$, by \cite[Lemma 2.4]{BFSV} $\psi_1 u \in H^1(\Omega_-;\cc^2) $ and in consequence  $ u \in H^1(\Omega_-;\cc^2)$.
\end{proof}
\begin{remark}\label{extTr} We mention that the same arguments as in the proof of Lemma \ref{tracemap} combined with \cite[Corollary~4.6]{BHSLS24} yield that 
the Dirichlet trace operators $\gamma_D^\pm$ extend for any $s\in[0,1]$ to continuous maps 
\[
\gamma_D^\pm:H^s_{\sigma, A}(\Omega_\pm; \mathbb{C}^2)\rightarrow H^{s-\frac{1}{2}}(\Sigma; \mathbb{C}^2). 
\] 
\end{remark}

We finish this part with the following density result and its consequence.
\begin{proposition}\label{density_result} $C^\infty_0(\overline{\Omega_\pm};\cc^2)$ is dense in $H_{\sigma, A}(\Omega_\pm;\cc^2)$.
\end{proposition}
\begin{proof} We only give the proof for $H_{\sigma, A}(\Omega_-;\cc^2)$, the same arguments yields the result for $H_{\sigma, A}(\Omega_+;\cc^2)$. The proof will follow by showing that if $u\in H_{\sigma, A}(\Omega_-;\cc^2)$ fulfills 
\begin{align}\label{density1}
   \langle u,v\rangle_{L^2(\Omega_-;\cc^2)} + \langle \sigma\cdot\nabla_A u, \sigma\cdot\nabla_A v\rangle_{L^2(\Omega_-;\cc^2)}=0 \quad \text{for all } v\in C^\infty_0(\overline{\Omega_-};\cc^2),
\end{align}
then $u=0$. To this end, let $U=\sigma\cdot\nabla_A u \in L^2(\Omega_-;\cc^2)$, and remark that \eqref{density1} implies that $\sigma\cdot\nabla_A U=-u$ holds in $\mathcal{D}^\prime(\Omega_-;\cc^2)$ and then in $L^2(\Omega_-;\cc^2)$. Denote by $U_0$ and $u_0$ the extensions by zero to $L^2(\rr^2;\cc^2)$ of $U$ and $u$, respectively. Then, for any $w\in C^\infty_0(\Omega_-;\cc^2)$ there holds
\begin{align*}
\langle\sigma\cdot\nabla_A U_0, w\rangle_{\mathcal{D}^\prime(\rr^2;\cc^2),\mathcal{D}(\rr^2;\cc^2)}&=\langle U_0, \sigma\cdot\nabla_A  w\rangle_{L^2(\rr^2;\cc^2)}=\langle U, \sigma\cdot\nabla_A  w\rangle_{L^2(\Omega_-;\cc^2)}\\
&=\langle - u, \sigma\cdot\nabla_A  w\rangle_{L^2(\Omega_-;\cc^2)}=\langle - u_0, \sigma\cdot\nabla_A  w\rangle_{L^2(\Omega_-;\cc^2)},
\end{align*}
which yields $\sigma\cdot\nabla_A U_0=-u_0$ in $L^2(\rr^2;\cc^2)$, and thus $U_0\in\dom D_0\subset H^1(\rr^2;\cc^2)$. Since $\gamma_D^\Sigma U_0=\gamma_D^+U= \gamma_D^-U_0=0$ it follows that $U\in H_{\sigma, A}(\Omega_-;\cc^2)\cap  H^1_0  (\Omega_-;\cc^2)$. \ From this it follows that there is a sequence $U_n\in C^\infty_0(\Omega_-;\cc^2)$ with $U_n \longrightarrow U$ in $H_{\sigma, A}(\Omega_-;\cc^2)$, in particular, we have $\sigma\cdot\nabla_A U_n \longrightarrow \sigma\cdot\nabla_A U$ in $L^2(\Omega_-;\cc^2)$. Therefore, using that $\sigma\cdot\nabla_A U=-u$ holds in $L^2(\Omega_-;\cc^2)$ and integrating by parts,  we get 
\begin{align*}
   \langle \sigma\cdot\nabla_A u, \sigma\cdot\nabla_A v\rangle_{L^2(\Omega_-;\cc^2)}&= \langle U, \sigma\cdot\nabla_A v\rangle_{L^2(\Omega_-;\cc^2)}=\lim\limits_{n\rightarrow \infty}\langle  U_n, \sigma\cdot\nabla_Av\rangle_{L^2(\Omega_-;\cc^2)}\\
   &=-\lim\limits_{n\rightarrow \infty}\langle \sigma\cdot\nabla_A U_n, v\rangle_{L^2(\Omega_-;\cc^2)}=-\langle \sigma\cdot\nabla_A U, v\rangle_{L^2(\Omega_-;\cc^2)}\\
   &=\langle u, v\rangle_{L^2(\Omega_-;\cc^2)},
\end{align*}
which plugged in \eqref{density1} yields $u=0$.
\end{proof}
A straightforward application of the Green formula combined with Proposition \ref{density_result} yields: 
\begin{lemma}\label{Green} For any $u_\pm,v_\pm\in  H^1_{\sigma, A}(\Omega_\pm;\cc^2)$ there holds 
      \begin{align*}
    \langle -i\sigma\cdot\nabla_A u_\pm, v_\pm \rangle_{L^2(\Omega_\pm;\cc^2)}=\langle u_\pm, -i\sigma\cdot\nabla_A v_\pm \rangle_{L^2(\Omega_\pm;\cc^2)}+\langle \mp i\sigma\cdot\nu\gamma_D^\pm u_\pm, \gamma_D^\pm v_\pm \rangle_{L^2(\Sigma;\cc^2)}.
\end{align*}
\end{lemma}

\subsection{Potential and boundary integral operators associated with $D_0$}

In what follows, it will be convenient to use the identification
\[
H^s_{\sigma, A}(\rr^2\setminus\Sigma;\cc^2)\simeq H^s_{\sigma, A}(\Omega_+;\cc^2)\oplus H^s_{\sigma, A}(\Omega_-;\cc^2),
\quad
u\simeq (u_+,u_-),
\]
with $u_\pm$ being the restriction of $u$ in $\Omega_\pm$, as well as analogous identifications for $H^s(\rr^2\setminus\Sigma;\cc^2)$ and $H_\sigma(\rr^2\setminus\Sigma;\cc^2)$.

We introduce the potential operators $\Phi_z: L^2(\Sigma; \mathbb{C}^2) \rightarrow L^2(\rr^2; \mathbb{C}^2)$ defined for all $z\in\rho(D_0)\setminus\{-m\}$ by
\begin{align} \label{def_Phi_z}
	\Phi_z   g(x) &:= \int_\Sigma \var_z (x,y) g(y) \,\ddd s(y). 
\end{align}
As $\var_z$ is a fundamental solution of $D_0-z$, it is clear that $(D_0-z)\Phi_z=0$ in $\rr^2\setminus\Sigma$. The following proposition gives some preliminary properties of $\Phi_z$.
\begin{proposition}\label{Basic_Phi}For any $z\in\rho(D_0)\setminus\{-m\}$ the operator $\Phi_z$ is well defined and bounded. Moreover, the following holds:
    \begin{itemize}
	  \item[$\textup{(i)}$]  $\Phi_z$ extends to a bounded bijective operator
	\begin{equation*}
	\begin{split}
	\Phi_z :H^{-\frac{1}{2}}(\Sigma;\mathbb{C}^2) \to \{ u\in H_{\sigma, A}(\rr^2\setminus\Sigma;\mathbb{C}^2): (D_0-z) u=0 \text{ in } \rr^2\setminus\Sigma\}.
	\end{split}
	\end{equation*}
	\item[$\textup{(ii)}$] The adjoint $\Phi_z^*: L^2(\rr^2; \mathbb{C}^2) \to  L^2(\Sigma; \mathbb{C}^2)$ of $\Phi_z$ acts on $u \in L^2(\rr^2; \mathbb{C}^2)$ as
	\begin{equation} \label{adj_Phi}
	  \Phi_{z}^* u(x) = \int_{\rr^2} \var_{\overline{ z}}(x,y) u(y) \, \ddd y, \qquad x \in \Sigma,
	\end{equation}
	and $\Phi_z^*$ gives rise to a bounded operator $\Phi_z^*: L^2(\rr^2; \mathbb{C}^2) \rightarrow H^{\frac{1}{2}}(\Sigma; \mathbb{C}^2)$. In particular, $\Phi_z^*$ is compact from $ L^2(\rr^2; \mathbb{C}^2)$ to $L^2(\Sigma; \mathbb{C}^2)$.
	\end{itemize}
\end{proposition}

\begin{proof} Fix $z\in\rho(D_0)$ and define the  operator 
           \[
		\widetilde{\Phi_{z}}^{\prime}=\gamma_D^\Sigma (D_0-\overline z)^{-1}: L^2(\rr^2; \mathbb{C}^2) \rightarrow H^{\frac{1}{2}}(\Sigma; \mathbb{C}^2),
		\]
      which is well-defined, bounded, and surjective. We then define the map $\widetilde{\Phi_{z}}$ as the anti-dual operator
		\[
		\widetilde{\Phi_{z}}=\left(\gamma_D^\Sigma (D_0-\overline z)^{-1}\right)^{\prime}: H^{-\frac{1}{2}}(\Sigma; \mathbb{C}^2)  \rightarrow L^2(\rr^2; \mathbb{C}^2).
		\]
        Let $g\in L^2(\Sigma; \mathbb{C}^2)$ and $u\in L^2(\rr^2; \mathbb{C}^2)$, then the symmetry property \eqref{var_ast} together with Fubini's theorem give
		\begin{align}\label{Dualcr}
        \begin{split}
			\langle \widetilde{\Phi_{z}}g , u \rangle_{L^2(\rr^2; \mathbb{C}^2)}&=\langle  g, \widetilde{\Phi_{z}}^{\prime} u \rangle_{L^2(\Sigma; \mathbb{C}^2)}\\
			&=\langle  g, \gamma_D^\Sigma (D_0-\Bar z)^{-1} u \rangle_{L^2(\Sigma; \mathbb{C}^2)}\\
			&=  \int_\Sigma \Big\langle g(x), \int_{\rr^2} \var_{\Bar{z}}(x,y)u(y)\mathrm{d} y \Big\rangle_{\cc^2}\,\mathrm{d}s(x)\\
			&=\int_{\rr^2} \Big\langle \int_\Sigma \overline{\var_{\Bar{z}}(x,y)}g(x)\,\mathrm{d}s(x), u(y)\Big\rangle_{\cc^2}\,\mathrm{d}y\\
			&=\int_{\rr^2} \Big\langle \int_\Sigma \var_{z}(y,x)g(x)\, \mathrm{d}s(x), u(y)\Big\rangle_{\cc^2}\,\mathrm{d}y\\
            &=\langle \Phi_{z}g , u \rangle_{L^2(\rr^2; \mathbb{C}^2)}.
            \end{split}
		\end{align}
Therefore, $\widetilde{\Phi_{z}}$ is an extension of $\Phi_z$ and thus $\Phi_z: L^2(\Sigma; \mathbb{C}^2) \rightarrow L^2(\rr^2; \mathbb{C}^2)$ is well defined and bounded for all $z\in\rho(D_0)$ which proves the first statement. Note also that the second and third equalities in \eqref{Dualcr} give the integral representation \eqref{adj_Phi} for the adjoint operator $\Phi_z^\ast$ and show that $\Phi_z^\ast=\widetilde{\Phi_{z}}^{\prime}=\gamma_D^\Sigma (D_0-\bar z)^{-1}$. Hence, $\Phi_z^*: L^2(\rr^2; \mathbb{C}^2) \rightarrow H^{\frac{1}{2}}(\Sigma; \mathbb{C}^2)$ is bounded, and since $H^{\frac{1}{2}}(\Sigma; \mathbb{C}^2)$ is compactly embedded in $L^2(\Sigma; \mathbb{C}^2)$ it follows that $\Phi_z^*$ is compact from $ L^2(\rr^2; \mathbb{C}^2)$ to $L^2(\Sigma; \mathbb{C}^2)$, which completes the proof of (ii). To complete the proof of (i),  consider the linear operator
\begin{equation*}
  \begin{split}
    \mathcal T v = D v, \quad    \dom \mathcal T =\{v\in\dom D_0: \gamma_D^\Sigma v=0 \},
  \end{split}
\end{equation*}
 By Lemma \ref{tracemap} one has $u_-\in H^1(\Omega_-;\cc^2)$ for any $u=u_+\oplus u_-\in\dom \mathcal T$. Then one easily checks that $\mathcal T$ is a densely defined, closed, and symmetric operator, and that 
 \begin{equation*}
  \begin{split}
    \mathcal T^\ast u = (Du_+)\oplus(D u_-),\quad \dom \mathcal T^\ast =\{u=u_+\oplus u_-\in H_{\sigma, A}(\rr^2\setminus\Sigma;\cc^2) \}.
  \end{split}
\end{equation*}
Since $D_0$ is self-adjoint and $\mathcal T$ is the restriction of $D_0$ to $\{u\in\dom D_0: \gamma_D^\Sigma u=0 \}$, it follows that 
 \begin{align*}
     \ran (\mathcal T-\Bar{z})&= \{u\in L^2(\rr^2;\cc^2):  (D_0-\Bar{z})^{-1}u\in\dom D_0 \text{ and } \gamma_D^\Sigma (D_0-\Bar{z})^{-1}u=0 \},\\
     &= \{u\in L^2(\rr^2;\cc^2):  \widetilde{\Phi_{z}}^{\prime}u=0 \}= \ker \widetilde{\Phi_{z}}^{\prime},
 \end{align*}
where the surjectivity of the trace map $\gamma_D^\Sigma$ was used in the second equality. As  $\ker \widetilde{\Phi_{z}}^{\prime}=(\ran  \widetilde{\Phi_{z}})^{\bot}$ and  $\ker(\mathcal T^\ast-z)=(\ran(\mathcal T-\Bar{z}))^\bot$ is a closed subspace,  the closed range theorem entails that $\widetilde{\Phi_z}$ is bounded and bijective from $H^{-\frac{1}{2}}(\Sigma;\mathbb{C}^2)$ onto $\ker(\mathcal T^\ast-z)$. Assertion (i) follows by setting $\widetilde{\Phi_z}g=\Phi_z g$ for $g\in H^{-\frac{1}{2}}(\Sigma;\mathbb{C}^2)$.        
\end{proof}

In the sequel, it will be convenient to extend the definition of $\Phi_z$ and $\Phi_{\Bar{z}}^\ast$ to the point  $z=-m$.  The following lemma gives some useful identities and properties of $\Phi_z$ and $\Phi_z^\ast$ allowing this.

\begin{lemma}\label{Phi_identities} 
\begin{enumerate}
\item For $z,\zeta\in\rho(D_0)\setminus\{-m\}$ one has $\Phi_{\Bar{z}}^\ast \Phi_{\zeta}= \Phi_{\Bar\zeta}^\ast\Phi_{z}$ and 
\begin{align}\label{Phi_identity}
\begin{split}
    \Phi_z-\Phi_{\zeta}=(z-\zeta)(D_0 -z)^{-1}\Phi_{\zeta}\quad\text{and} \quad  \Phi_z^\ast-\Phi_{\zeta}^\ast&=(\Bar{z}-\Bar{\zeta})\Phi_{\zeta}^\ast(D_0 -\Bar{z})^{-1}.
    \end{split}
\end{align} 
\item In particular, for a fixed  $\zeta\in\rho(D_0)\cap\rr$ with $\zeta\neq-m$,  if one sets 
\[
 \Phi_{-m}=\Phi_{\zeta}-(m+\zeta)(D_0 +m)^{-1}\Phi_{\zeta},
 \]
  then the statement of Proposition \ref{Basic_Phi} still holds true for $z=-m$ with 
\[
 \Phi_{-m}^\ast= \Phi_{\zeta}^\ast-(m+ \zeta)\Phi_{\zeta}^\ast(D_0 +m)^{-1}.
 \]
  Furthermore, the operator-valued functions  
\begin{align*} 
z \mapsto \Phi_z \in \mathcal L(H^{-\frac{1}{2}}(\Sigma;\mathbb{C}^2), L^2(\rr^2;\mathbb{C}^2) ),\quad   z \mapsto \Phi_{\Bar{z}}^\ast \in \mathcal L(L^2(\rr^2;\mathbb{C}^2), H^{\frac{1}{2}}(\Sigma;\mathbb{C}^2)) ,
\end{align*}
are holomorphic on $\rho(D_0)$.
\end{enumerate}
\end{lemma}
\begin{proof}Fix $z,\zeta\in\rho(D_0)\setminus\{-m\}$ and recall the resolvent identity 
\begin{align*}
(D_0 -\Bar{z})^{-1}-(D_0 -\zeta)^{-1}&=(\Bar{z}-\Bar{\zeta})(D_0 -\Bar{\zeta})^{-1}(D_0 -\Bar{z})^{-1}\\
&\equiv(\Bar{z}-\Bar{\zeta})(D_0 -\Bar{z})^{-1}(D_0 -\Bar{\zeta})^{-1}.
\end{align*}
Applying the trace operator gives 
\begin{align*}
    \Phi_z^\ast-\Phi_{\zeta}^\ast&=(\Bar{z}-\Bar{\zeta})\Phi_{\zeta}^\ast(D_0 -\Bar{z})^{-1}=(\Bar{z}-\Bar{\zeta})\Phi_{z}^\ast(D_0 -\Bar{\zeta})^{-1},
\end{align*}
and taking the adjoint yields the identity   
\[
\Phi_z-\Phi_{\zeta}=(z-\zeta)(D_0 -z)^{-1}\Phi_{\zeta}=(z-\zeta)(D_0 -\zeta)^{-1}\Phi_{z},
\]
and applying again the trace gives the identity $\Phi_{\Bar{z}}^\ast \Phi_{\zeta}= \Phi_{\Bar\zeta}^\ast\Phi_{z}$.

Finally, since the map $\rho(D_0) \ni z \mapsto (D_0 -z)^{-1}$ is holomorphic,  by definition the maps $z\mapsto \Phi_z$ and $z\mapsto \Phi_{\Bar{z}}^\ast $ are holomorphic in $\rho(D_0)$.
\end{proof}

For the subsequent analysis, we define the singular integral operator $H_\Sigma :\mathit{L}^2(\Sigma)  \longrightarrow \mathit{L}^2(\Sigma)$ by
\begin{align}\label{Cauchy2}
		\begin{split}
			H_{\Sigma}h(x)&:= \dfrac{i}{2\pi}\lim\limits_{\rho\searrow 0}\int_{\substack{y\in\Sigma\\ |x-y|>\rho}}\dfrac{h(y)}{(x_1-y_1) +i(x_2-y_2)}\,\mathrm{d}s(y),
		\end{split}
\end{align}
and the single layer boundary operator  $S(\mu) :\mathit{L}^2(\Sigma)  \longrightarrow \mathit{L}^2(\Sigma)$ by

\begin{align}\label{SL}
S(\mu)h(x):= 	\begin{cases}
		-\dfrac{1}{2\pi}\int_{\Sigma}\log|x-y|h(y)\,\mathrm{d}s(y), \quad \text{if } &\, \mu=0,\\
		\dfrac{1}{2\pi}\int_{\Sigma}K_0(-i\sqrt{\mu}|x-y|)h(y)\,\mathrm{d}s(y), \quad \text{if }& \, \mu\in(\infty,0).
	\end{cases}
\end{align}
Here $K_0$ is the modified Bessel function of the second kind. The next lemma states some known results on the pseudodifferential properties of $H_\Sigma$ and $S(\mu)$.  In the statement and what follows, we use the notation $\Op \cS^n_l(\Sigma)$ to denote the class of classical pseudodifferential operators of order $n$ acting on sections of $\Sigma\times \cc^l$.

Set 
\begin{align}\label{Def T}
T\,:=\, \text{the multiplication operator by $ t_1+ i t_2$ in $L^2(\Sigma),$}
\end{align}
where $t=(t_1,t_2)^{\top}:=(-\nu_2,\nu_1)^{\top}$ is the tangent vector field on $\Sigma$. 
\begin{lemma}\label{H properties} Let  $H_\Sigma$ and $S(\mu)$ be as above. Then the following hold:
\begin{itemize}
\item[(i)] For any $\mu\in(\infty,0]$, $S(\mu)$ is a  classical elliptic pseudodifferential operator of order $-1$, and its  principal symbol is given by $p_{S(\mu)}(x,\xi)= \frac{1}{2|\xi|}$ for all $\xi\in T^\ast \Sigma\setminus\{0\}$. In addition, $S(\mu)$ is self-adjoint and non-negative for $\mu\in(\infty,0)$.

 \item[(ii)] $H_\Sigma\in \Op \cS^0_1(\Sigma)$  with principal symbol
\[
	p_{H_\Sigma}(x,\xi) = -\frac{1}{2}T(x) \mathrm{sgn}(\xi):= -\frac{1}{2}T(x)\frac{\xi}{|\xi|},\quad \text{ for all } \xi\in T^\ast \Sigma\setminus\{0\}.
\]
	Moreover, it holds that $(H_\Sigma\overline{T})^2=\frac{1}{4} I$.
\end{itemize}

\end{lemma}

\begin{proof}  That $S(\mu)$ is a  classical elliptic pseudodifferential operator of order $-1$ with  $p_{S_\Sigma}(x,\xi)= \frac{1}{2|\xi|}$  is  proved  in \cite[§ 11.4.1]{HW}, see also \cite[Proposition 11.2]{T1}. Now,  that $S(\mu)$  is self-adjoint and non-negative for $\mu\in(\infty,0)$ is shown in \cite[Sec. 2]{BCP}, and this gives (i).

To proof of (ii) relies on the observation that $\mathcal H=-H_\Sigma \overline{T}$ corresponds to the Hilbert transform on $\Sigma$ (see e.g. \cite[Sec.~5]{bell}). As is well-known \cite[§ 11.3.1, p. 693]{HW},  this operator is a classical pseudodifferential operator of order zero with principal symbol given by $p_{\mathcal H}(x,\xi)= \frac{1}{2}\mathrm{sgn}(\xi)$, and satisfy $4 \mathcal H^2 =I$, see e.g. \cite[§  27]{musk}. 
\end{proof}

 We next consider the singular integral operator $\mathscr{C}_{z}: L^2(\Sigma,\cc^2)\longrightarrow L^2(\Sigma,\cc^2)$ defined by 
\begin{align}\label{SIO}
  \mathscr{C}_{z}g(x) &= \pv \int_{\Sigma} \varphi_{z}(x,y)g(y)\,\mathrm{d}s(y),\quad x\in\Sigma,
\end{align}
if $z\in\rho(D_0)\setminus\{-m\}$, and in view of Lemma \ref{Phi_identities}, for a fixed  $\zeta\in\rr\setminus(\spec (D_0)\cap \{-m\}$,  we also admit the value $z=-m$ by setting
 \begin{equation}\label{C_-m}
 \mathscr C_{-m}=\mathscr C_{\zeta}- (m+\zeta)\Phi^\ast_{-m}\Phi_{\zeta}.
 \end{equation}
The basic properties of $\mathscr{C}_{z}$  are summarized in the following proposition. 
\begin{proposition}\label{Basic_C_z} For any  $z\in\rho(D_0)$ we have 
\begin{align}\label{IdenC}
\mathscr C_z= \Theta + (z I_2+m\sigma_3)(S(0) \otimes I_2)  + T_z, \quad 
\Theta:=   \begin{pmatrix} 
	       0& H_\Sigma\\ 
	H^*_\Sigma & 0 \, 
	   \end{pmatrix},
\end{align}
where $T_z$ is bounded from $H^{-\frac{1}{2}}(\Sigma;\cc^2)$ to $H^{\frac{1}{2}}(\Sigma;\cc^2)$. In particular,  the following hold:
\begin{itemize}
			\item[$\textup{(i)}$] For any $s\in[-\frac{1}{2},\frac{1}{2}]$ the map $\mathscr{C}_{z}$ gives rise to a bounded operator  $\mathscr{C}_{z}: H^s(\Sigma,\cc^2)\longrightarrow H^s(\Sigma,\cc^2)$. 
            \item[$\textup{(ii)}$] For any $s\in[-\frac{1}{2},\frac{1}{2})$ there is a compact operator $K_z$ in $H^s(\Sigma,\cc^2)$ such that $\mathscr{C}_z^2= \frac{1}{4}I + K_z$ as an operator in $H^s(\Sigma,\cc^2)$. 
            \item[$\textup{(iii)}$] If $z\in\rho(D_0)\cap\rr$, then $\mathscr{C}_{z}$ is self-adjoint in $L^2(\Sigma,\cc^2)$.
		\end{itemize}
\end{proposition}
\begin{proof}   Since $\Phi^\ast_{\Bar{z}}\Phi_{-m}$ is bounded from $H^{-\frac{1}{2}}(\Sigma;\cc^2)$ to $H^{\frac{1}{2}}(\Sigma;\cc^2)$ by Proposition \ref{Basic_Phi} (i)-(ii), it gives rise to a compact operator in $H^s(\Sigma,\cc^2)$ for each $s\in[-\frac{1}{2},\frac{1}{2})$. Hence, by definition of $\mathscr C_{-m}$  in \eqref{C_-m}  it suffices to show the result for $z\neq -m$. So, fix $z\in \rho(D_0)\setminus\{-m\}$ and recall the decomposition of the kernel $\varphi_z$ from Lemma \ref{lemma:varphi_asymp}. Then we have
\begin{align*}
\mathscr C_z= \Theta + (z I_2+m\sigma_3)(S(0) \otimes I_2) + T_z, 
\end{align*}
where $\Theta$ is the integral operator in $L^2(\Sigma;\cc^2)$ defined by
\[
\Theta g(x)=\mathrm{p.v } \int_\Sigma\theta(x-y)g(y)\mathrm{d}s(y), \quad g\in L^2(\Sigma;\cc^2),
\]
 with kernel 
\begin{align*}
     \theta (x) \;:=\; \,\frac{i}{2\pi}\,\sigma\!\cdot\!\frac{x}{|x|^2}=\frac{i}{2\pi}\left(\begin{array}{cc}
0& \dfrac{1}{x_1+ix_2} \\
\dfrac{1}{x_1-ix_2} & 0
\end{array}\right),
\end{align*}
and  $T_z$ is an integral operator with a  continuous kernel in a neighborhood of the origin and $C^\infty$ outside it. By standard arguments which treat layer potential operators as pseudodifferential operators (cf. \cite[Chap. 7, Sec. 11]{T1}), it follows that $T_z$ gives a rise to a pseudodifferential operator of order $-2$, i.e., $T_z \in \Op \cS^{-2}_{2}(\Sigma)$. In particular, $T_z: H^{s}(\Sigma;\cc^2) \rightarrow H^{s+2}(\Sigma;\cc^2)$ is bounded for any $s\in\rr$.

A straightforward computation shows that 
 \begin{equation*}
\Theta=   \begin{pmatrix} 
	       0& H_\Sigma\\ 
	H^*_\Sigma & 0 \, 
	   \end{pmatrix}.
\end{equation*} 
This proves the identity \eqref{IdenC} and that $\mathscr{C}_{z}\in \Op \cS^0_2(\Sigma)$ with principal symbol
\begin{align}\label{symbol_C}
    p_{\mathscr{C}_{z}}(x,\xi)= -\frac{1}{2} \begin{pmatrix} 
	       0& T(x) \mathrm{sgn}(\xi) \\ 
	\overline{T}(x) \mathrm{sgn}(\xi)  & 0 \, 
	   \end{pmatrix}.
\end{align}
Thus, $\mathscr{C}_{z}$ maps continuously $H^s(\Sigma;\cc^2)$ into itself for any $s\in\rr$ and this gives (i). Since $ p_{\mathscr{C}_{z}}(x,\xi)^2= \frac{1}{4} I_2$ it follows that $\mathscr{C}_z^2= \frac{1}{4}I$ modulo  $\Op \cS^{-1}_2(\Sigma)$ which shows (ii).  Finally, the fact that $\mathscr{C}_z$ is self-adjoint in $L^2(\Sigma;\cc^2)$  follows directly from the identities \eqref{var_ast}.
\end{proof}

We next state a jump formula describing the trace of $\Phi_z$ defined in \eqref{def_Phi_z}  in terms of the singular integral operator $\mathscr{C}_{z}$.

\begin{proposition}\label{Further_Phi_z} Let $z\in\rho(D_0)$ and let  $\mathscr{C}_{z}$ and $\Phi_z$ be as above. Then, for any $g \in H^{-\frac{1}{2}}(\Sigma;\mathbb{C}^2)$ one has
			\begin{align}\label{eq_C_z_formula}
				\gamma_D^\pm(\Phi_z g)_\pm =    \mathscr{C}_zg   \,\mp \frac{i}{2}  (\sigma \cdot \nu)g. 
			\end{align}
 In particular, $\Phi_z$ gives rise to a bounded operator
 \[
 \Phi_z: H^{\frac{1}{2}}(\Sigma;\mathbb{C}^2)\longrightarrow H^1(\Omega_+;\mathbb{C}^2) \oplus H^1_{\sigma, A}(\Omega_-;\mathbb{C}^2),
 \]
 and for any $s\in[-\frac{1}{2},\frac{1}{2}]$ the maps 
\begin{align*} 
z \mapsto \Phi_z \in \mathcal L(H^{\frac{1}{2}}(\Sigma;\mathbb{C}^2), H^1(\Omega_+;\mathbb{C}^2) \oplus H^1_{\sigma, A}(\Omega_-;\mathbb{C}^2)),\quad   z \mapsto \mathscr C_z\in \mathcal L(H^{s}(\Sigma;\mathbb{C}^2)) 
\end{align*}
are holomorphic in $\rho(D_0)$.
\end{proposition}
\begin{proof}  Thanks to the decomposition  \eqref{IdenC},  the identity \eqref{eq_C_z_formula} is proved exactly as in \cite[Proposition 3.4]{BHOP20} for $z\in\rho(D_0)\setminus\{-m\}$. The same holds for $z=-m$ thanks to \eqref{Phi_identity}. 

Since $\nu$ is smooth and $\mathscr{C}_z$ is bounded in $H^{\frac{1}{2}}(\Sigma;\cc^2)$, it follows that $\gamma_D^\pm(\Phi_zg)_\pm\in H^{\frac{1}{2}}(\Sigma;\cc^2)$ for any $g\in H^{\frac{1}{2}}(\Sigma;\cc^2)$. Therefore, Lemma \ref{tracemap} implies that  $\Phi_zg\in H^1(\Omega_+;\mathbb{C}^2) \oplus H^1_{\sigma, A}(\Omega_-;\mathbb{C}^2)$. Finally, the holomorphic properties of $\Phi_z$ and $\mathscr C_z$ are consequence of the holomorphic property of $(D_0-z)^{-1}$.
\end{proof}

\section{The Landau--Dirac operator with $\delta$-interactions: Basic properties}\label{secLDD}
In this section, we discuss the basic spectral properties of the operator defined formally by
	\[
	D_{\epsilon, \tau}= D+ (\epsilon I_2 + \tau \sigma_3)\delta_\Sigma.
	\]
    Here and in the sequel, $\epsilon,\tau: \Sigma\to \rr$ are $C^1$-smooth functions such that 
     \begin{align}\label{assumption1}
     |\epsilon(s)|\neq |\tau(s)|\,\,\, \text{for any }\, s\in\Sigma.
     \end{align}
     As in the non-magnetic case, the operator $D_{\epsilon,\tau}$ can be defined rigorously via a transmission condition. Namely, for $u\in H^1(\Omega_+,\cc^2)\oplus H^1_{\sigma, A}(\Omega_-,\cc^2)$ a straightforward computation in the sense of distributions gives 
    \[
    D_{\epsilon,\tau}u= (D u_+)\oplus(D u_-)+ \left[\frac{1}{2}(\epsilon I_2+\tau\sigma_3)(\gamma_D^+u_+ +\gamma_D^-u_-) +i\sigma\cdot\nu(\gamma_D^+u_+-\gamma_D^-u_-)\right]\delta_\Sigma.
    \]
    Hence, if we assume that the transmission condition
\begin{align}\label{TC}
    \frac{1}{2}(\epsilon I_2+\tau\sigma_3)(\gamma_D^+u_+ +\gamma_D^-u_-) +i\sigma\cdot\nu(\gamma_D^+u_+-\gamma_D^-u_-)=0,
 \end{align}
is satisfied in $H^{\frac{1}{2}}(\Sigma;\cc^2)$, then $D_{\epsilon,\tau}u\in L^2(\rr^2;\cc^2)$. This leads us to define $D_{\epsilon,\tau}$ as follows
	\begin{equation} \label{def_A_eta_tau}
		\begin{split}
			D_{\epsilon,\tau}u = (D u_+) \oplus (D u_-), \\
			\dom D_{\epsilon,\tau} = \Big\{ u = u_+ \oplus u_- \in  H^1(\Omega_+,\cc^2)&\oplus H^1_{\sigma, A}(\Omega_-,\cc^2): \\
            &\, \eqref{TC} \text{ holds in } H^{\frac{1}{2}}(\Sigma;\cc^2) \Big\}.
		\end{split}
	\end{equation}

\begin{lemma}\label{Closable} Let $D_{\epsilon,\tau}$ be defined as above. Then, $D_{\epsilon,\tau}$ is densely defined and symmetric. In particular, $D_{\epsilon,\tau}$ is closable. 
\end{lemma}
Before presenting the proof of Lemma \ref{Closable}, we provide an alternative formulation of the transmission condition \eqref{TC} in the following lemma. 
\begin{lemma}\label{TC_alter} Let $\epsilon,\tau\in C^1(\Sigma)$ satisfy \eqref{assumption1} and set 
\begin{equation*}\label{Projection}
    P_\pm=\frac{1}{2}\left(I_2 \mp \frac{i}{2}\sigma\cdot\nu(\epsilon I_2+\tau\sigma_3) \right).
\end{equation*}
Then, the following is true:  
\begin{itemize}
    \item[(i)] For those $s\in\Sigma$ such that   $\epsilon^2(s)-\tau^2(s)=-4$ we have that  $P_+$ and $P_-$ are projectors (i.e., $P_+ +P_-=I_2$ and $P_\pm^2= P_\pm$) and satisfy 
    \begin{align}\label{P_pm}
   P_\pm^\ast=\frac{1}{2}\left(I_2 \pm \frac{i}{2}\sigma\cdot\nu(\epsilon I_2-\tau\sigma_3) \right)\quad\text{and}\quad   P_\pm^\ast(i\sigma\cdot\nu)=(i\sigma\cdot\nu)P_\mp.
    \end{align}
  In particular, if  $\epsilon^2-\tau^2\equiv-4$ on $\Sigma$, then  for any $u=u_+\oplus u_-\in\dom D_{\epsilon,\tau}$ the transmission condition \eqref{TC} is equivalent to
    \[
    P_+ \gamma_D^+u_+=0 \quad\text{and}\quad  P_- \gamma_D^-u_-=0,
    \]
  and $D_{\epsilon,\tau}$ decomposes into the direct sum $D^{\Omega_+}_{\epsilon,\tau}\oplus D^{\Omega_-}_{\epsilon,\tau}$, where $D^{\Omega_\pm}_{\epsilon,\tau}$ are the Dirac operators defined by
    \begin{align}\label{H_Omega}
		\begin{split}
			D^{\Omega_\pm}_{\epsilon,\tau}u_\pm =& D u_\pm , \\
			\dom D^{\Omega_+}_{\epsilon,\tau} =& \{ v \in  H^1(\Omega_+,\cc^2): P_+ \gamma_D^+v=0  \},\\
            \dom D^{\Omega_-}_{\epsilon,\tau} =& \{ v \in H^1_{\sigma, A}(\Omega_-,\cc^2): P_- \gamma_D^-v=0  \}.
		\end{split}
	\end{align}
    \item[(ii)] For those $s\in\Sigma$ such that   $\epsilon^2(s)-\tau^2(s)\neq-4$ we have that   $P_+$ and $P_-$ are invertible with 
    \[
    (P_\pm)^{-1}=\frac{16}{\epsilon^2-\tau^2+4}P_{\mp}.
    \]
    In this case, for any $u=u_+\oplus u_-\in\dom D_{\epsilon,\tau}$ the transmission condition \eqref{TC} reads as follows
    \[
  Q_\pm \gamma_D^\mp u_\mp(s)=  \gamma_D^\pm u_\pm(s),
    \]
    where $Q_\pm$ are the invertible matrices
    \[
    Q_\pm= \frac{4}{\epsilon^2-\tau^2+4}\left(\frac{4 -\epsilon^2+\tau^2}{4}I_2 \pm i\sigma\cdot\nu(\epsilon I_2+\tau\sigma_3)) \right).
    \]
\end{itemize}
\end{lemma}
\begin{proof} The proof follows by straightforward computations and proceeds along the same lines as those in the proof of \cite[Lemma 4.1]{BHOP20} where the case of constant interactions $\epsilon$ and $\tau$ is treated, we omit the details.
\end{proof}

We can now give the proof of Lemma \ref{Closable}.\\

\textbf{Proof of Lemma \ref{Closable}.} Since inclusions $C^\infty_0(\rr^2\setminus\Sigma;\cc^2)\subset \dom D_{\epsilon,\tau}\subset L^2(\rr^2;\cc^2)$ holds and $C^\infty_0(\rr^2\setminus\Sigma;\cc^2)$ is a dense subspace of $L^2(\rr^2;\cc^2)$ it follows that $D_{\epsilon,\tau}$ is a densely defined. We now check that $D_{\epsilon,\tau}$ is symmetric. Let  $u,v\in\dom D_{\epsilon,\tau}$, and set $J_{-4}:=\{s\in \Sigma: \epsilon^2(s)-\tau^2(s)=-4\}$ and $J^c_{-4}:=\Sigma\setminus J_{-4}$. 
Then a simple computation using Lemma \ref{Green} and the fact that the vector potential $A$ is smooth gives 
\begin{align*}
   & \langle D_{\epsilon,\tau} u, v\rangle_{L^2(\rr^2;\cc^2)}-\langle u, D_{\epsilon,\tau} v \rangle_{L^2(\rr^2;\cc^2)}
   \\=&\langle - i\sigma\cdot\nu\gamma_D^+ u_+, \gamma_D^+ v_+ \rangle_{L^2(\Sigma;\cc^2)}+ \langle  i\sigma\cdot\nu\gamma_D^- u_-, \gamma_D^- v_- \rangle_{L^2(\Sigma;\cc^2)}\\
    =& \int_{J_{-4}}  \langle - i\sigma\cdot\nu\gamma_D^+ u_+,{ \gamma_D^+ v_+}\rangle_{\mathbb C^2}\,\mathrm{d}s +\int_{J^c_{-4}} \langle-i\sigma\cdot\nu\gamma_D^+ u_+,{ \gamma_D^+ v_+}\rangle_{\mathbb C^2}\, \mathrm ds\\
    +&  \int_{J_{-4}} \langle i\sigma\cdot\nu\gamma_D^- u_-,{ \gamma_D^- v_-} \rangle_{\mathbb C^2}\,\mathrm ds + \int_{J^c_{-4}}  \langle i\sigma\cdot\nu\gamma_D^- u_-,{ \gamma_D^- v_-} \rangle_{\mathbb C^2}\,\mathrm{d}s.
\end{align*}
Now, by Lemma \ref{TC_alter}  we have $Q_+ \gamma_D^- u_-=  \gamma_D^+ u_+$ and $Q_+ \gamma_D^- v_-=  \gamma_D^+ v_+$  when $\epsilon^2(s)-\tau^2(s)\neq-4$, then 
    \begin{align*}
   &\int_{J^c_{-4}} \langle- i\sigma\cdot\nu\gamma_D^+ u_+,{ \gamma_D^+ v_+} \rangle_{\mathbb C^2}\,\mathrm{d}s+\int_{J^c_{-4}}  \langle i\sigma\cdot\nu\gamma_D^- u_-,{ \gamma_D^- v_-} \rangle_{\mathbb C^2}\,\mathrm{d}s\\
    =&  \int_{J^c_{-4}}  \langle- i\sigma\cdot\nu Q_+\gamma_D^- u_-, {Q_+\gamma_D^- v_-}  \rangle_{\mathbb C^2}\,\mathrm{d}s + \int_{J^c_{-4}}  \langle i\sigma\cdot\nu\gamma_D^- u_-, \gamma_D^- v_-\rangle_{\mathbb C^2}\,\mathrm{d}s \\
    =& \int_{J^c_{-4}} \langle i[ (\sigma\cdot\nu) - Q_+^\ast(\sigma\cdot\nu)Q_+]\gamma_D^- u_-, \gamma_D^- v_- \rangle_{\mathbb C^2}\,\mathrm{d}s=0.   
\end{align*}
 where we used the relation  $Q_+^\ast(\sigma\cdot\nu)Q_+ =\sigma\cdot\nu$. Similarly, in the case where $\epsilon^2(s)-\tau^2(s)=-4$, Lemma \ref{TC_alter} implies that $P_\pm\gamma_D^\pm u_\pm=P_\pm\gamma_D^\pm v_\pm=0$. Combining this with the property \eqref{P_pm} yields
 \begin{align*}
&\int_{J_{-4}} \langle- i\sigma\cdot\nu\gamma_D^+ u_+,{ \gamma_D^+ v_+} \rangle_{\mathbb C^2}\,\mathrm{d}s+\int_{J_{-4}}  \langle i\sigma\cdot\nu\gamma_D^- u_-,{ \gamma_D^- v_-} \rangle_{\mathbb C^2}\,\mathrm{d}s\\ 
=&\int_{J_{-4}} \langle- i\sigma\cdot\nu P_-\gamma_D^+ u_+,P_-{ \gamma_D^+ v_+} \rangle_{\mathbb C^2}\,\mathrm{d}s+\int_{J_{-4}}  \langle i\sigma\cdot\nu P_+\gamma_D^- u_-,P_+{ \gamma_D^- v_-} \rangle_{\mathbb C^2}\,\mathrm{d}s\\
=&\int_{J_{-4}} \langle \underbrace{P_-^\ast(- i\sigma\cdot\nu) P_-}_{=0}\gamma_D^+ u_+, \gamma_D^+ v_+ \rangle_{\mathbb C^2}\,\mathrm{d}s+\int_{J_{-4}}  \underbrace{\langle P_+^\ast( i\sigma\cdot\nu) P_+}_{=0}\gamma_D^- u_-, \gamma_D^- v_- \rangle_{\mathbb C^2}\,\mathrm{d}s=0.
\end{align*} 
Therefore $D_{\epsilon,\tau}$ is symmetric. As any symmetric operator on a Hilbert space with a dense domain of definition always admits
a closure, this completes the proof.
\qed

\section{Self-adjointness and other spectral properties of $D_{\epsilon,\tau}$: Non-critical case}\label{S4}
For $z\in\rho(D_0)\setminus\{-m\}$, we define the auxiliary operators   in $L^2(\Sigma;\mathbb C^2)$
\begin{align}\label{def_Pi}
    \Pi_z= \frac{1}{\epsilon^2-\tau^2}(\epsilon I_2-\tau\sigma_3)+ \mathscr{C}_z \quad\text{and}\quad \underline{ \Pi_z}= \frac{1}{\epsilon^2-\tau^2}(\epsilon I_2+\tau\sigma_3)- \mathscr{C}_z,
\end{align}
where $\mathscr{C}_z$ is defined by \eqref{SIO}. Since $\epsilon$ and $\tau$ are $C^1$-smooth, by Proposition \ref{Basic_C_z} it is clear that $\Pi_z$ and $\underline{ \Pi_z}$ are well defined and bounded in $H^s(\Sigma;\cc^2)$ for any $s\in[-\frac{1}{2},\frac{1}{2}]$. 
\begin{lemma}\label{Pi_Fredholm} Assume that $\epsilon(s)^2-\tau(s)^2\neq 4$ for all $s\in\Sigma$. Then, for any $z\in\rho(D_0)$ 
\begin{enumerate}
\item[(i)]  The operator $\Pi_z$ is Fredholm with index zero in $L^2(\Sigma;\cc^2)$,
\item[(ii)] $\ker (D_{\epsilon,\tau}-z)=\Phi_z\{ g\in H^{\frac{1}{2}}(\Sigma;\cc^2): g\in\ker \Pi_z\}$. In particular, it holds that $\mathrm{dim}\ker (D_{\epsilon,\tau}-z)=\mathrm{dim}\ker \Pi_z$.
\end{enumerate}
\end{lemma}
\begin{proof} Let us show (i). By  definition of $\mathscr{C}_z$ is suffices to prove the Fredholmness for a fixed $z\in\rho(D_0)\setminus\{-m\}$. Using that $(\epsilon I_2 + \tau\sigma_{3})(\epsilon I_2 - \tau\sigma_3)=\epsilon^{2} - \tau^{2}\neq 0$, we compute 
\begin{align*}
 \underline{ \Pi_z}\Pi_z&= \frac{1}{\epsilon^2-\tau^2}-\mathscr{C}_z^2+  [\frac{\epsilon}{\epsilon^2-\tau^2}, \mathscr{C}_z]+ \{\frac{\tau}{\epsilon^2-\tau^2}\sigma_3,\mathscr{C}_z\},
\end{align*}
and 
\begin{align*}
	\Pi_z\underline{ \Pi_z}= \frac{1}{\epsilon^2-\tau^2}-\mathscr{C}_z^2+  [\mathscr{C}_z, \frac{\epsilon}{\epsilon^2-\tau^2}]+ \{\frac{\tau}{\epsilon^2-\tau^2}\sigma_3,\mathscr{C}_z\},
\end{align*}
where $[M,N ]:=MN-NM$ and $\{M,N \}:=MN+NM$ denote the commutator and anticommutator of $M$ and $N$. Recall that $\mathscr{C}_{z}$ is a zero-order pseudodifferential operator and its principal symbol $p_{\mathscr{C}_z}(x,\xi)$ is given by \eqref{symbol_C}. Since $\epsilon$ and $\tau$ are $C^1$-smooth, the multiplications by $\epsilon$ and $\tau$ give rise to pseudodifferential operators of order zero, respectively, with symbols given by
\[
p_{\frac{\epsilon}{\epsilon^2-\tau^2}}(x,\xi)=\frac{\epsilon}{\epsilon^2-\tau^2}(x) \quad\text{and} \quad p_{\frac{\tau}{\epsilon^2-\tau^2}}(x,\xi)=\frac{\tau}{\epsilon^2-\tau^2}(x).
\]
Since $\epsilon$ and $\tau$ are scalar valued, it follows that $[\frac{\epsilon}{\epsilon^2-\tau^2}, \mathscr{C}_z]$ and $\{\frac{\tau}{\epsilon^2-\tau^2}\sigma_3,\mathscr{C}_{z}\}$ are pseudodifferential operators of order $-1$ and $0$, respectively. Moreover, we have 
\begin{align*}
      p_{\frac{\tau}{\epsilon^2-\tau^2}}(x,\xi)\left(\sigma_3 p_{\mathscr{C}_{z}}(x,\xi)+  p_{\mathscr{C}_{z}}(x,\xi)\sigma_3\right)= 0,
\end{align*}
which means that  $\{\frac{\tau}{\epsilon^2-\tau^2}\sigma_3,\mathscr{C}_z\}$ is of order -1. Therefore, $[\frac{\epsilon}{\epsilon^2-\tau^2}, \mathscr{C}_z]$ and
 $\{ \frac{\tau}{\epsilon^2-\tau^2}\sigma_3,\mathscr{C}_z\}$ are compact in $L^2(\Sigma;\cc^2)$. Thus, Proposition \ref{Basic_C_z}-(ii) implies that 
\begin{align}\label{inverse_Pi}\begin{split}
\underline{\Pi_{z}}\Pi_{z}= \frac{4-\epsilon^2+\tau^2}{4(\epsilon^2-\tau^2)}+ T_{1} \quad\text{and}\quad \Pi_{z}\underline{\Pi_{z}}= \frac{4-\epsilon^2+\tau^2}{4(\epsilon^2-\tau^2)}+ T_{2} ,
\end{split}
\end{align}
for some compact operators $T_j$ in $L^2(\Sigma;\cc^2)$. Since $\epsilon^2-\tau^2\neq 4$, by Theorems 1.50 and 1.51 in \cite{Aiena}, we deduce that $\Pi_{z}$ and $\underline{\Pi_{z}}$ are Fredholm operators.  Moreover, since $\mathscr{C}_z$ is self-adjoint for $z\in\rho(D_0)\cap\rr$ (Proposition \ref{Basic_C_z}-(iii)),  $\Pi_{z}$ shares this property and hence has  index zero.

Now, let us prove (ii). Let $0\neq u\in \ker (D_{\epsilon,\tau}-z)$. By Proposition \ref{Basic_Phi}-(i) we have  $u=\Phi_z g$ with $g\in H^{-\frac{1}{2}}(\Sigma;\cc^2)$. Using Proposition \ref{Further_Phi_z} we obtain 
\begin{align}\label{Traceeigen}
\begin{split}
\frac{1}{2}(\epsilon I_2+\tau\sigma_3)(\gamma_D^+ u_+ +\gamma_D^- u_-)&=(\epsilon I_2+\tau\sigma_3)\mathscr{C}_zg,\\
 i\sigma\cdot\nu(\gamma_D^+ u_+ -\gamma_D^- u_-)&=g.
 \end{split}
\end{align}
From the last equality we see that $g\in H^{\frac{1}{2}}(\Sigma;\cc^2)$. Moreover, as $u$ satisfies the transmission condition and $(\epsilon I_2+\tau\sigma_3)$ is invertible, it follows that  
\[
\Pi_zg= (\epsilon I_2+\tau\sigma_3)^{-1}\left( I_2+ (\epsilon I_2+\tau\sigma_3)\mathscr{C}_z\right)g=0,
\]
which shows that $g\in\ker \Pi_z$ and implies the inclusion $\subset$. Conversely, if $0\neq g\in\ker \Pi_z$, then $u=\Phi_z g \in \dom D_{\epsilon,\tau}$ and satisfies $(D_{\epsilon,\tau}-z)u=0$ by Propositions \ref{Basic_Phi}-(i) and \ref{Further_Phi_z} and the computation carried out above. Since $\Phi_z$ is injective, we have $u\neq 0$, and this proves the reverse inclusion $\supset$ and the equality $\mathrm{dim}\ker (D_{\epsilon,\tau}-z)=\mathrm{dim}\ker \Pi_z$. 

\end{proof}
The following theorem concerns the self-adjointness and basic spectral properties of $D_{\epsilon,\tau}$ for the  non-critical case $\epsilon^2-\tau^2\neq 4$. 
\begin{theorem}\label{main_th1}Let $\epsilon,\tau\in C^1(\Sigma)$ satisfy $(\epsilon^2-\tau^2)(\Sigma) \cap\{0, 4\}=\emptyset$. Then for any $z \in \rho(D_0)\cap\rho(D_{\epsilon, \tau})$ the operator $ \Pi_{z}$ is boundedly invertible in $L^2(\Sigma; \mathbb{C}^2)$, and the resolvent formula
			\begin{equation}\label{Krein}
				(D_{\epsilon, \tau} - z)^{-1} = (D_0 - z)^{-1} - \Phi_{z} \Pi_{z}^{-1} \Phi_{\Bar{z}}^*,
\end{equation}				
holds. Moreover, $D_{\epsilon,\tau}$ is self-adjoint and there holds 
 \[
 \spece (D_{\epsilon,\tau})=\spece (D_0)=\spec (D_0).
 \]			
\end{theorem}
\begin{proof}
 Let $z \in \rho(D_0)\cap \rho(D_{\epsilon, \tau})$, then assertion (ii) in Lemma \ref{Pi_Fredholm} implies that $\ker \Pi_z=\{0\}$. Since $\Pi_z$ is Fredholm with zero index by Lemma \ref{Pi_Fredholm}-(i), it follows that $\ran \Pi_z=L^2(\Sigma; \mathbb{C}^2)$ and thus it is boundedly invertible in $L^2(\Sigma; \mathbb{C}^2)$. 
 
Denote by $R(z)$ the bounded linear operator
\begin{align*}\label{Kresolvent1}
	R(z)= (D_0-z)^{-1} - \Phi_{z}\Pi_{z}^{-1} \Phi^{\ast}_{\overline{z}}:\ L^2(\rr^2,\cc^2)\to L^2(\rr^2,\cc^2),
\end{align*}
and let us prove that $R(z)=(D_{\epsilon,\tau}-z)^{-1}$. Given $u\in L^2(\rr^2,\cc^2)$, using Propositions \ref{Basic_Phi}-(i) and \ref{Further_Phi_z}, and the fact that $\Pi_z^{-1}: \ H^{\frac{1}{2}}(\Sigma^2,\cc^2)\to H^{\frac{1}{2}}(\Sigma^2,\cc^2)$, we obtain that
\[
v:=R(z)u\in H^1(\Omega_+,\cc^2)\oplus H^1_{\sigma, A}(\Omega_-,\cc^2),
\]
and that
\begin{align*}
\frac{1}{2}(\gamma_D^+v_+ + \gamma_D^-v_-)&=\gamma_D^\Sigma(D_0-z)^{-1} u -\mathscr{C}_z \Pi_{z}^{-1} \Phi^*_{\overline{z}}u
=\Phi^*_{\Bar{z}}v-\mathscr{C}_z \Pi_{z}^{-1} \Phi^*_{\Bar{z}}u,\\
\gamma_D^+ v_+ -\gamma_D^- v_- &=i\sigma\cdot\nu\Pi_{z} ^{-1} \Phi^*_{\Bar{z}}u.
\end{align*}
Hence, using the identities $(\epsilon I_2 + \tau\sigma_{3})(\epsilon I_2 - \tau\sigma_3)=\epsilon^{2} - \tau^{2}$  and $\Phi^*_{\Bar{z}}=\gamma_D^\Sigma(D_0-z)^{-1}$, we obtain
\begin{align*}
\frac{1}{2}(\epsilon I_2+\tau\sigma_3)(\gamma_D^+ v_+ +\gamma_D^- v_-)+ &i\sigma\cdot\nu(\gamma_D^+ v_+ -\gamma_D^- v_-)\\
&=(\epsilon I_2 + \tau\sigma_{3})\Big[I-\left(\frac{1}{\epsilon^{2} - \tau^{2}}(\epsilon I_2 - \tau\sigma_3)+\mathscr{C}_z \right)\Pi_z^{-1})\Big] \Phi^*_{\Bar{z}}u\\
&=(\epsilon I_2 + \tau\sigma_{3})(I-I)\Phi^*_{\Bar{z}}u=0,
\end{align*}
which  shows that $v$ satisfies the transmission condition \eqref{TC}, and therefore it  is in $\dom D_{\epsilon,\tau}$. Since  $(D_0-z)\Phi_{z}\Pi_{z}^{-1} \Phi^{\ast}_{\overline{z}}=0$ in $\rr^2\setminus\Sigma$, it follows that
\begin{align*}
(D_{\epsilon,\tau}-z)v=(D-z)R(z)u=(D-z)(D_0-z)^{-1}u=u,
\end{align*}
and thus $R(z)=(D_{\epsilon,\tau}-z)^{-1}$, which shows the resolvent formula \eqref{Krein}.

We now turn to the proof of  the self-adjointness.  Let $z\in \cc\setminus\rr$, then $\ker (D_{\epsilon,\tau}-z)=\{0\}$ because $D_{\epsilon, \tau}$ is symmetric, and in view of the resolvent formula \eqref{Krein} we have $\ran (D_{\epsilon,\tau}-z)=L^2(\rr^2;\cc^2)$. From this we conclude that $D_{\epsilon,\tau}$ is self-adjoint. 

Finally,  since $\Phi_z^*$ is compact from $ L^2(\rr^2; \mathbb{C}^2)$ to $L^2(\Sigma; \mathbb{C}^2)$ by Proposition \ref{Basic_Phi}-(ii) and $\Phi_z\Pi_z^{-1}: L^2(\Sigma; \mathbb{C}^2) \rightarrow L^2(\rr^2; \mathbb{C}^2)$ is bounded,  it follows that $\Phi_z\Pi_z^{-1} \Phi_{\bar z}^*$ is a compact operator on $ L^2(\rr^2; \mathbb{C}^2)$.  By  \eqref{Krein} we deduce that the resolvent difference $(D_{\epsilon, \tau} - z)^{-1} - (D_0 - z)^{-1}$ is a compact operator in $ L^2(\rr^2; \mathbb{C}^2)$, and we conclude by  Weyl's theorem that $\spece (D_{\epsilon,\tau})=\spece (D_0)$. \end{proof}

The following proposition gives additional properties of the spectrum of $D_{\epsilon, \tau}$ in the confining case $\epsilon^2-\tau^2=-4$. In this situation, thanks to Theorem \ref{main_th1} and Lemma \ref{TC_alter},  we have the decomposition $D_{\epsilon, \tau}=D_{\epsilon,\tau}^{\Omega_+}\oplus D_{\epsilon,\tau}^{\Omega_-}$ and the operators $D_{\epsilon,\tau}^{\Omega_\pm}$ defined in \eqref{H_Omega} are self-adjoint.  

\begin{proposition}\label{Spec_H} Let $\epsilon,\tau\in C^1(\Sigma)$ satisfy $\epsilon^2(s)-\tau^2(s)= -4$ for all $s\in\Sigma$. Then, the spectrum of $D_{\epsilon,\tau}^{\Omega_+}$ is purely discrete and  $\spece(D_{\epsilon,\tau}^{\Omega_-})=\spece(D_0)$. In addition, if $\tau(s)\geq2$ 
holds for each $s\in\Sigma$, then we have
\[
\specd(D_{\epsilon,\tau}^{\Omega_+})\subset \rr\setminus [-m,m] \quad \text{and}\quad \specd(D_{\epsilon,\tau}^{\Omega_-})\subset \rr\setminus [-m,m).
\]
\end{proposition}
\begin{proof} Clearly $D_{\epsilon,\tau}^{\Omega_+}$ has a compact resolvent implying that its spectrum is purely discrete. Consequently, Theorem \ref{main_th1} yields that $\spece(D_{\epsilon,\tau}^{\Omega_-})=\spece(D_{\epsilon,\tau})=\spece(D_0)$.

We are now going to prove that $\specd(D_{\epsilon,\tau})\subset \rr\setminus [-m,m)$ holds for $\tau\geq2$ .  For any $u=u_+\oplus u_-\in\dom D_{\epsilon,\tau}$ we have 
\[
\|D_{\epsilon,\tau}u \|^2_{L^2(\rr^2;\cc^2)}= \|D_{\epsilon,\tau}^{\Omega_+}u_+ \|^2_{L^2(\Omega_+;\cc^2)}+\|D_{\epsilon,\tau}^{\Omega_-} u_- \|^2_{L^2(\Omega_-;\cc^2)},
\]
 and using Lemma \ref{Green}, we compute    
 \begin{align}\label{Quadratic1}
 \begin{split}
 \|D_{\epsilon,\tau}^{\Omega_\pm}u_\pm \|^2_{L^2(\Omega_\pm;\cc^2)}=& \langle D u_\pm, D u_\pm\rangle_{L^2(\Omega_\pm;\cc^2)}\\
 =&  \|\sigma\cdot\nabla_Au_\pm \|^2_{L^2(\Omega_\pm;\cc^2)} + m^2  \| u_\pm \|^2_{L^2(\Omega_\pm;\cc^2)}\\
 &+m\langle \mp i\sigma\cdot\nu \gamma_D^\pm u_\pm, \sigma_3 \gamma_D^\pm u_\pm\rangle_{L^2(\Sigma;\cc^2)} .  
 \end{split}
 \end{align}
By Lemma \ref{TC_alter}-(i), we have 
 \begin{align}\label{Quadratic2}
 \gamma_D^\pm u_\pm= \pm\frac{i}{2}\sigma\cdot\nu(\epsilon I_2+\tau\sigma_3)\gamma_D^\pm u_\pm.
 \end{align}
 Thus, if write  $u_\pm =(v_\pm,w_\pm)^\top$ then  \eqref{Quadratic2} is equivalent to
 \[
 \mp i\sigma\cdot\nu \gamma_D^\pm u_\pm= \begin{pmatrix}
     \frac{\epsilon+\tau}{2} & 0\\
     0 & \frac{\epsilon-\tau}{2}
 \end{pmatrix} \gamma_D^\pm u_\pm  = \frac{1}{2} \begin{pmatrix}(\epsilon+\tau)\gamma_D^\pm v_\pm\\
(\epsilon-\tau) \gamma_D^\pm w_\pm
 \end{pmatrix}.
 \]
 This implies 
 \begin{align*}
 \langle \mp i\sigma\cdot\nu \gamma_D^\pm u_\pm, \sigma_3 \gamma_D^\pm u_\pm\rangle_{L^2(\Sigma;\cc^2)}= \frac{1}{2}\int_\Sigma \left[  (\epsilon+\tau) |v_\pm |^2 + (\tau-\epsilon) |w_\pm |^2 \right]\mathrm{d}s.
 \end{align*}
  Plugging this into  \eqref{Quadratic1}  yields 
 \begin{align}\label{Quadratic3}
 \begin{split}
 \|D_{\epsilon,\tau}u \|^2_{L^2(\rr^2;\cc^2)} =&  \|\sigma\cdot\nabla_A u_+ \|^2_{L^2(\Omega_+;\cc^2)} + \|\sigma\cdot\nabla_A u_- \|^2_{L^2(\Omega_-;\cc^2)} + m^2 \|u \|^2_{L^2(\rr^2;\cc^2)} \\
 &+ \frac{m}{2} \int_\Sigma (\epsilon+\tau)(|v_+ |^2 + |v_- |^2) \mathrm{d} s   + \frac{m}{2} \int_\Sigma (\tau-\epsilon)( |w_+ |^2 + |w_- |^2)\mathrm{d}s.
 \end{split}
 \end{align}
Since $\tau \geq2$ it follows that $\tau>|\epsilon|$, and  \eqref{Quadratic3} entails that
 \begin{align*}
 \|D_{\epsilon,\tau}u \|^2_{L^2(\rr^2;\cc^2)} \geq m^2 \|u \|^2_{L^2(\rr^2;\cc^2)}
 \end{align*}
 showing that $\spec(D_{\epsilon,\tau})\subset \rr\setminus (-m,m)$. In view of \eqref{Quadratic3}, if $-m$ is an eigenvalue of $D_{\epsilon,\tau}$, then its associated eigenfunction $u$ satisfies $\gamma_D^\pm u_\pm=0$. However, $u=\Phi_{-m} g$ with $g=i\sigma\cdot\nu(\gamma_D^+ u_+-\gamma_D^- u_-)$ by \eqref{Traceeigen}, and we conclude that $u$ vanishes identically on $\rr^2$. Therefore, $-m\notin\spec(D_{\epsilon,\tau})$ which leads to $\specd(D_{\epsilon,\tau})\subset \rr\setminus [-m,m)$.  

Similarly, if  $m\in \specd (D_{\epsilon,\tau}^{\Omega_+})$ then by  \eqref{Quadratic3} its associated eigenfunction $u_+\in \dom D_{\epsilon,\tau}^{\Omega_+}\setminus\{0\}$ satisfies 
\begin{align*}
\gamma_D^+ u_+=0 \,\,\text{on }\,\Sigma \quad\text{and}\quad \sigma\cdot\nabla_A u_+=0 \,\,\text{in }\,\Omega_+.
\end{align*}
Hence, $u_+ =(v_+,w_+)^\top\in H^1_0(\Omega_+;\cc^2)$ and applying $\sigma\cdot\nabla_A$ to the second equality above gives
\begin{align*}
0=(\sigma\cdot\nabla_A)^2 u_+\;\;
&=\; \begin{pmatrix} -\nabla_A^2 - b & 0 \\ 0 & -\nabla_A^2 + b \end{pmatrix}\begin{pmatrix} v_+\\
  w_+
 \end{pmatrix}.
\end{align*}
This contradicts the fact that the infimum of the spectrum of the magnetic Schr\"odinger operator with Dirichlet boundary condition is larger than the lowest Landau level $b$. Therefore,  $m\notin \specd (D_{\epsilon,\tau}^{\Omega_+})$ and the proposition is proved.
\end{proof}
\begin{remark} In the confining regime $\epsilon^2-\tau^2=-4$, the condition  $\pm \tau \geq 2$ on $\Sigma$ is equivalent to $\pm(\epsilon+\tau) >0$ on $\Sigma$. As illustrated in Theorem \ref{Cluster},  if $\tau \leq -2$ on $\Sigma$, then $D^{\Omega_-}_{\epsilon,\tau}$ admits infinitely many eigenvalues accumulating to the point $m$ from below. 
\end{remark}

\section{Eigenvalues Asymptotics of $D_{\epsilon,\tau}$ in the non-critical case}\label{S5}

In this section we are interested in the influence of the $\delta$-interactions $(\epsilon I_2 + \tau \sigma_3)\delta_\Sigma$ on the distribution of the eigenvalues of $D_{\epsilon,\tau}$  near the Landau--Dirac levels $\mu_q$, $q \in \zz$. It is expected that under the influence of the perturbation, each infinite multiplicity eigenvalue $\mu_q$ of  $D_0$ will split into an infinity of discrete eigenvalues (counted with their multiplicities) that accumulate at $\mu_q$. To study the distribution of the eigenvalues close to a fixed Landau--Dirac level $\mu_q$, we fix a sequence of reals numbers $\{\alpha_q\}_{q\in \zz}$ such that $ \mu_q < \alpha_q  < \mu_{q+1}$  and we introduce the functions $ {\mathcal N}_{+}^q$ (resp.  ${\mathcal N}_{-}^q$) that count the eigenvalues of $D_{\epsilon,\tau}$ in the interval $(\mu_q + \lambda, \alpha_q)$ (resp. $(\alpha_{q-1}, \mu_q - \lambda)$) for $\lambda >0$ close to $0$. 

More generally, if  $E_{\mathcal A}(\omega)$ denotes the spectral projection of the self-adjoint operator $\mathcal A$  associated with the Borel set $\omega$ and if $\omega$ does not intersect the essential spectrum of $\mathcal A$, we define the eigenvalue counting function 
\[
\mathcal N(\omega; \mathcal A):= {\rm Tr}E_{\mathcal A}(\omega).
\]
In our case, for $\lambda >0$ we set
 \[
 {\mathcal N}_{+}^q(\lambda):=\mathcal N((\mu_q + \lambda, \alpha_q); D_{\epsilon,\tau}),\quad  {\mathcal N}_{-}^q(\lambda):=\mathcal N((\alpha_{q-1},\mu_q-\lambda); D_{\epsilon,\tau}).
 \]
If $\mathcal A$ is also compact, for $s>0$  we also use the notation  
 \[
 n_\pm(s;\mathcal A)=\mathcal N((s,\infty); \pm \mathcal A).
 \]
Note that the behaviour of ${\mathcal N}_{\pm}^q(\lambda)$, as $\lambda \searrow 0$, is independent of the choice of $\{\alpha_q\}_q$ because $\mu_q$ is the only possible accumulation point in $(\alpha_{q-1}, \alpha_q)$. Thus either the number of eigenvalues near $\mu_q$ is finite (i.e. ${\mathcal N}_{\pm}^q(\lambda)=O(1)$), or  infinitely many eigenvalues accumulate at $\mu_q$ (from the right when ${\mathcal N}_{+}^q(\lambda) \rightarrow + \infty$ or the left when ${\mathcal N}_{-}^q(\lambda) \rightarrow + \infty$).

As we shall see below, contrary to previous results in similar models (\cite{ RaWa02, FiPu06, MeRo07, PushRoz07,  RoTa08, Pe09, PuRo11, GoKaPe16, BHOP20, BM}), accumulation can occur on  the left of the Landau levels even if the matrix $\epsilon I_2 + \tau \sigma_3$ is positive definite.
 In spite of this, the rate of accumulation is still described in terms of the {\em logarithmic capacity} of $\Sigma$, defined by $\capa(\Sigma) : = e^{-{\mathcal I}(\Sigma)}$ where
\[
{\mathcal I}(\Sigma) : = \inf_{\mu \in \gM(\Sigma)}  \int_{\Sigma} \int_{\Sigma} \ln{|x-y|^{-1}} d\mu(x) d\mu(y)
\]    
and $\mathfrak{M}(\Sigma)$ is the set of compactly supported probability measures on $\Sigma$.

We can now state the main result of this section concerning the accumulation of eigenvalues to the Landau--Dirac levels. 
\begin{theorem}\label{Cluster}
Assume that $\Sigma $ is a $C^\infty$ curve and   $\epsilon,\tau$ are in $C^1(\Sigma)$ such that $(\epsilon^2-\tau^2)(\Sigma) \cap \{0,4\} = \emptyset$. 
Then the effective function determining the side of the spectral accumulation is 
\[
V_1:=\frac{4(\epsilon+\tau)}{4-(\epsilon^2-\tau^2)},
\]
in the sense that if $\pm V_1 >0$ on $\Sigma$ then, as $\lambda \searrow 0$, ${\mathcal N}_{\mp}^q(\lambda)= O(1)$ and 
\begin{align}\label{Asymp}
 {\mathcal N}_{\pm}^q(\lambda)=\frac{|\ln{\lambda}|}{\ln|\ln(\lambda)|} + \frac{|\ln{\lambda}|\ln(\ln|\ln(\lambda)|)}{\ln|\ln(\lambda)|^2}  +   \frac{|\ln{\lambda}|}{\ln|\ln(\lambda)|^2}\big(\mathfrak C(\Sigma)+o(1)\big),
\end{align}
 with $\mathfrak C(\Sigma)  = 1+  \ln{\left(\frac{b}2\,{\rm Cap}(\Sigma)^2\right)}$.
\end{theorem} 

Before  giving the proof of this theorem, let us give the following consequence in the confining case $\epsilon^2-\tau^2=-4$ that is for $V_1=(\epsilon+\tau)/2$ whose sign is that of $\tau$. In this situation, we have the decomposition $D_{\epsilon, \tau}=D_{\epsilon,\tau}^{\Omega_+}\oplus D_{\epsilon,\tau}^{\Omega_-}$ and the operator $D_{\epsilon,\tau}^{\Omega_+}$ has purely discrete spectrum (see Proposition \ref{Spec_H}). Then the above result of spectral accumulation is necessarily carried by the exterior Hamiltonian $D_{\epsilon,\tau}^{\Omega_-}$ and it confirms a conjecture of \cite{BM} 
in the particular - confining - case $(\epsilon, \tau)=(0, \pm 2)$ corresponding to  \textit{infinite mass boundary conditions} (see Remark \ref{rqconf}).

\begin{corollary}\label{ClusConf} Let $\epsilon,\tau\in C^1(\Sigma)$ satisfy $\epsilon^2(s)-\tau^2(s)= -4$ for all $s\in\Sigma$. Then eigenvalues of the operator $D_{\epsilon,\tau}^{\Omega_-}$ accumulate above (resp. below) each Landau level $\mu_q$, $q \in \zz$,  when $\tau (s) \geq 2$ (resp. $\tau  (s) \leq - 2$)  for all $s\in\Sigma$. On the other side the spectrum is finite. 
\end{corollary}

Since the operators $D_0$ and $D_{\epsilon,\tau}$ are not defined on the same domain, for the proof of the above Theorem, it will be easier to work on their inverse (or their resolvent at a fixed real number $\zeta_0$) and to exploit that
 \begin{align*}
 {\mathcal N}_{+}^q(\lambda)&=\mathcal N\Big((\frac{1}{\alpha_q- \zeta_0},\frac{1}{\mu_q- \zeta_0+\lambda});{(D_{\epsilon,\tau}- \zeta_0)^{-1}}\Big)
 , \\
 {\mathcal N}_{-}^q(\lambda)&=\mathcal N\Big((\frac{1}{\mu_q- \zeta_0-\lambda}, \frac{1}{\alpha_{q-1} -\zeta_0}); (D_{\epsilon,\tau}- \zeta_0)^{-1}\Big).
 \end{align*}
For $\zeta_0 \in \rho(D_0) \cap \rho(D_{\epsilon,\tau})\cap \rr\setminus\{-m\}$ define the bounded operator
\[
\Upsilon:=(D_{\epsilon,\tau} - \zeta_0)^{-1}-(D_0 - \zeta_0)^{-1} = -\Phi_{\zeta_0}\,\Pi_{\zeta_0}^{-1}\,\Phi_{\zeta_0}^{*},
\]
where the second equality follows from Theorem \ref{main_th1}. Then the proof of Theorem \ref{Cluster} consists of two main steps: First we reduce the problem to the counting function (near $0$) of a compact operator, the Toeplitz-type operator $ \mathcal P_q \Upsilon \mathcal P_q$ where ${\mathcal P}_{q}$ is  the  infinite dimensional orthogonal projection onto Ker$(D_0-\mu_{q})$. Second we study the eigenvalue distribution of this Toeplitz operator.

In the interest of simplifying the expressions, in what follows, we will assume that we can choose $\zeta_0=0$ and we omit the subscript for $\Phi_{\zeta_0}$, $\Pi_{\zeta_0}$: $\Phi:=\Phi_{\zeta_0}$, $\Pi:=\Pi_{\zeta_0}$.
Moreover we will give the proof of Theorem \ref{Cluster} only for $V_1>0$ and for $q\geq 0$. The same reasoning applies to the other cases.

\subsection{Reduction to a Toeplitz operator }\label{SEC51}

The first step consists to prove the following Lemma saying that the distribution of the eigenvalues near $\mu_q$ is governed by the compact operator $\mathcal P_q\Upsilon  \mathcal P_q$.

\begin{lemma}\label{le1} Let  $q\in \zz$  and set $\lambda_0:=\frac{\lambda}{(\mu_q+\lambda)\mu_q}$. Then  
for any $\delta>0$ we have that 
	\begin{align*}
	n_-(\lambda_0,\mathcal P_q(\Upsilon + \delta \Phi \Phi^{*})\mathcal P_q)+O(1)&\leq \mathcal N_+^q(\lambda)\leq  n_-(\lambda_0,\mathcal P_q(\Upsilon - \delta \Phi \Phi^{*})\mathcal P_q)+O(1),
\end{align*}
as $\lambda \searrow 0$.
\end{lemma}

This Lemma is obtained by exploiting the following spectral result.

\begin{lemma}\label{le1b} Let $\mu_0$ be in the essential spectrum of a self-adjoint bounded operator $\mathcal A_0$ such that for some $\beta< \mu_0$ we have $[\beta, \mu_0) \subset \rho(\mathcal A_0)$ (i.e. $\mathcal A_0$ has a spectral gap below $\mu_0$). 
\begin{itemize}
     \item[(i)] If $M$, $M_-$, $M_+$ are self-adjoint compact operators such that
    $M_- \leq M \leq M_+$, in the sense of forms, then uniformly w.r.t. $\lambda >0$, we have:
    \[
  \mathcal N([\beta, \mu_0 - \lambda); \mathcal A_0 +M_+) + O(1)  \leq \mathcal N([\beta, \mu_0 - \lambda); \mathcal A_0 +M) \leq \mathcal N([\beta, \mu_0 - \lambda); \mathcal A_0 +M_-) + O(1).
    \]
    \item[(ii)] If moreover, $\mu_0$ is an isolated eigenvalue of infinite multiplicity of $\mathcal A_0$ and if the associated eigenspace is stable by the compact operators $M_-$ and $M_+$, then as $\lambda \searrow 0$, 
    \[
  n_-(\lambda, \mathcal P_0 M_+ \mathcal P_0) + O(1)  \leq \mathcal N([\beta, \mu_0 - \lambda); \mathcal A_0 +M) \leq n_-(\lambda, \mathcal P_0 M_- \mathcal P_0) + O(1),
    \]
    where $\mathcal P_0$ denotes the orthogonal projection onto $\ker(\mathcal A_0 - \mu_0)$.
\end{itemize}

\end{lemma}
The first point of this lemma is a direct consequence of min-max principle when $\mu_0$ is the bottom of the essential spectrum. In more general settings it can be proved by using monotonicity properties for the index of a pair of spectral projections  (e.g. as in \cite[Section 4.1]{BM}). For the proof if (ii), it suffices to decompose $\mathcal A_0 +M_\pm$ into $E_0:=\ker(\mathcal A_0 -\mu_0)$ and its orthogonal. On $E_0^\perp$, the spectrum of $\mathcal A_0 +M_\pm$ is discrete in $[\beta, \mu_0]$, then 
\[ \mathcal N([\beta, \mu_0 - \lambda); \mathcal A_0 + M_\pm) =
 \mathcal N([\beta, \mu_0 - \lambda);  \mu_0 + \mathcal P_0 M_\pm \mathcal P_0 ) + O(1),
\]
and the result follows because 
\[
\mathcal N([\beta, \mu_0 - \lambda);  \mu_0 + \mathcal P_0 M_\pm \mathcal P_0 ) = \mathcal N([\beta-\mu_0, - \lambda);  \mathcal P_0 M_\pm \mathcal P_0 ) = n_-(\lambda, \mathcal P_0 M_\pm \mathcal P_0) + O(1).
\]

\begin{proof}[Proof of Lemma \ref{le1}]  Set $\mathcal P_q^{\perp}:= I -\mathcal P_q$.  By Theorem \ref{main_th1} we have that  $\Upsilon= -\Phi\,\Pi^{-1}\,\Phi^{*}$, 
then for all $\delta > 0$ and any $f\in L^2(\mathbb R^2;\cc^2)$, we have 
\begin{align*}
|\langle \mathcal P_q \Upsilon  \mathcal P_q^{\perp} f, f \rangle_{L^2(\mathbb R^2;\cc^2)}|
= &|\langle \Pi^{-1}\Phi^* \mathcal P_q^{\perp} f, \Phi^*\mathcal P_q f \rangle_{L^2(\Sigma;\cc^2)}|\\
\leq & \|\Pi^{-1} \Phi^* \mathcal P_q^{\perp} f\|_{L^2(\Sigma;\cc^2)} \, \|\Phi^* \mathcal P_qf \|_{L^2(\Sigma;\cc^2)} \\
\leq & \delta^{-1} \, \|\Pi^{-1} \Phi^* \mathcal P_q^{\perp} f\|_{L^2(\Sigma;\cc^2)}^2 + \delta \|\Phi^* \mathcal P_qf \|_{L^2(\Sigma;\cc^2)}^2\\
=&\delta^{-1} \langle \Pi^{-2}\Phi^* \mathcal P_q^{\perp} f, \Phi^*\mathcal P_q^{\perp} f \rangle_{L^2(\Sigma;\cc^2)}
+ \delta\langle \Phi^* \mathcal P_q f, \Phi^*\mathcal P_q f \rangle_{L^2(\Sigma;\cc^2)}.
\end{align*}
Writing
\[
\Upsilon = \mathcal P_q \Upsilon \mathcal P_q + \mathcal P_q^{\perp} \Upsilon \mathcal P_q^{\perp} + \mathcal P_q \Upsilon \mathcal P_q^{\perp} + \mathcal P_q^{\perp} \Upsilon \mathcal P_q,
\]
it follows, in the sense of forms, that  
\[
M_{-\delta} \leq \Upsilon \leq M_{+\delta} 
\]
with
\[M_{\pm \delta}:= \mathcal P_q( \Upsilon \pm \delta \Phi \Phi^{*})\mathcal P_q 
+ \mathcal P_q^{\perp}( \Upsilon \pm \delta^{-1} \Phi \Pi^{-2}\Phi^*)\mathcal P_q^{\perp} .
\]
Thanks to Proposition \ref{Basic_Phi}, the operators $M_{\pm \delta}$ are compact and we can apply Lemma \ref{le1b} for $\mathcal A_0= D_0^{-1} $, $\mu_0=\frac1{\mu_q}$, $M= \Upsilon$ and $M_\pm=M_{\pm \delta}$. We obtain:
\begin{align*}
	&n_-(\lambda_0,\mathcal P_q(\Upsilon + \delta \Phi \Phi^{*})\mathcal P_q)+O(1)\\
	\leq\, & \mathcal N \Big((\alpha_q^{-1},(\mu_q+\lambda)^{-1});{D_{\epsilon,\tau}^{-1}}\Big)\\ \leq & n_-(\lambda_0,\mathcal P_q(\Upsilon - \delta \Phi \Phi^{*})\mathcal P_q)+O(1),
\end{align*}
because $\mu_q^{-1} - (\mu_q+\lambda)^{-1} = \lambda_0 $, and the lemma follows.
\end{proof}

\subsection{Counting function of eigenvalues for the effective Hamiltonian}\label{SEC52}
   Lemma \ref{le1} entails that our {effective Hamiltonian} is  $-\mathcal P_q\Phi \Xi_{\delta} \Phi^{*}\mathcal P_q$  with $\Xi_{\delta}:= \Pi^{-1} -\delta I_2$ and  $\pm \delta >0$. 
Recalling  that $\Phi^{*}=\gamma^\Sigma_D\,\,D_{0}^{-1}$, the above effective Hamiltonian becomes formally:
\(
- \mu_q^{-2} \mathcal P_q (\gamma^\Sigma_D)^{*} \Xi_{\delta} \gamma^\Sigma_D\mathcal P_q
\). 
Moreover, from \eqref{inverse_Pi} it can be seen that
\begin{align*}
\Xi_{\delta}=\frac{4(\epsilon^2-\tau^2)}{4-(\epsilon^2-\tau^2)}
\begin{pmatrix}
\frac{\epsilon +\tau}{ \epsilon^2-\tau^2}& - H_\Sigma\\
- H_\Sigma^* & \frac{\epsilon -\tau}{\epsilon^2-\tau^2}
\end{pmatrix}
+S - \delta I_2=:  \begin{pmatrix}
V_1 - \delta & W_\Sigma\\
W^*_\Sigma &V_2 - \delta
\end{pmatrix}
+S,
\end{align*}
 where  \( H_\Sigma\) is the bounded operator defined in \eqref{Cauchy2} and $S$ is a compact operator in $L^2(\Sigma;\cc^2)$. Then inspired by \cite{BM}, the proof of Theorem \ref{Cluster} consists to prove that the distribution of the eigenvalues near $\mu_q$ is formally governed by $ p_{|q|}  (- V_1)  \delta_\Sigma p_{|q|}$, with $p_n$  the orthogonal projections onto $\ker(L -\Lambda_{n})$ where $\Lambda_n:=2nb$, $n \in \nn$ and
 \[ 
L:= a^* a = -\nabla_A^2  -b.
\]
Using the notation of \eqref{pi} we can see that  the orthogonal projections $\mathcal P_q$ are given by
\[
\mathcal{P}_q =
\begin{cases}
U_{FW}^*p_q\pi_+ U_{FW}\quad&\text{if }\,q\geq 0,\\
U_{FW}^*p_{|q|-1}\pi_- U_{FW}\quad&\text{if }\,q< 0,
\end{cases}
\]
where $U_{FW}$ is the unitary operator
\[
U_{FW}:=
\bigl(I + \sigma_3 (D_0 - m \sigma_3)|D_0 - m \sigma_3|^{-1} \bigr)\frac1{\sqrt 2} \sqrt{I - m |D_0|^{-1}}.
\]
For more details see \cite[Section 3]{BM}.

\begin{proof}[Proof of Theorem \ref{Cluster}]
We give the proof for $V_1>0$ on $\Sigma$ and for $q\geq 0$. As in \cite{BM}, the same arguments work for $q<0$.

Let $\tilde \pi_+: L^2(\mathbb R^2)\to L^2(\mathbb R^2;\mathbb C^2)$ be the operator defined by $\tilde \pi_+ f=(f,0)^\intercal$. Then it is obvious that the operator  $-\mathcal P_q(\Upsilon + \delta \Phi \Phi^{*})\mathcal P_q= \mathcal P_q\Phi \Xi_{\delta} \Phi^{*}\mathcal P_q$ has the same non-zero eigenvalues than the operator defined by the quadratic form  $Q$ in $L^{2}(\mathbb{R}^{2})$, 
\begin{align}\label{Def Q}
\begin{split}
Q[f]:=&\bigl\langle \Xi_{\delta}\Phi^{*}\mathcal{P}_q U_{FW}^*\tilde\pi_+ f,\;\Phi^{*}\mathcal{P}_q U_{FW}^*\tilde\pi_+ f \bigr\rangle_{L^2(\Sigma)}\\
=&\bigl\langle \Xi_{\delta} \Phi^{*}U_{FW}^*\pi_+p_qf ,\;\Phi^{*}U_{FW}^*\pi_+p_qf \bigr\rangle_{L^2(\Sigma)}.
\end{split}
\end{align}
Set $\mathcal A:=\begin{pmatrix}
V_{1} -\delta & W_\Sigma\\
W^*_\Sigma & V_{2} - \delta
\end{pmatrix}$ and define  the quadratic forms in $L^2(\mathbb R^2;\cc^2)$
\begin{align*}
Q_{\mathcal A}[f]:=&\bigl\langle \mathcal A \Phi^{*}U_{FW}^*\pi_+p_qf ,\;\Phi^{*}U_{FW}^*\pi_+p_qf \bigr\rangle_{L^2(\Sigma)},\\
Q_S[f]:=&\bigl\langle S\Phi^{*}U_{FW}^*\pi_+p_qf ,\;\Phi^{*}U_{FW}^*\pi_+p_qf \bigr\rangle_{L^2(\Sigma)}.
\end{align*}
So that $Q=Q_{\mathcal A}+Q_S$.

From the definition of $U_{FW}$, it is possible to show that there exist real numbers $c_1$ and $c_2$ such that
\begin{align*}
{U}_{FW}^* {p}_q \pi_+=&
\bigl(c_1 I - c_2 \sigma_3 (D_0 - m \sigma_3)\bigr)p_q \pi_+
=
\begin{pmatrix}
c_1p_q & 0 \\
c_2ap_q & 0
\end{pmatrix}.
\end{align*}
Therefore, recalling  that $\Phi^{*}=\gamma^\Sigma_D\,\,D_{0}^{-1}$ and $D_{0}^{-1}\mathcal{P}_q = \mu_q^{-1}\mathcal{P}_q$,  we obtain
\begin{align*}
Q_{\mathcal A}[f]
=&\alpha _1
\left\langle 
(V_1-\delta) \left.(p_q f)\right|_{\Sigma},
\left.(p_q f)\right|_{\Sigma}
\right\rangle_{L^2(\Sigma)}
+\alpha _2
\left\langle 
(V_2 -\delta) \left.(ap_q f)\right|_{\Sigma},
\left.(ap_q f)\right|_{\Sigma}
\right\rangle_{L^2(\Sigma)}\\
+&{\beta}
{\rm Re}
\left\langle 
W_\Sigma \left.(ap_q f)\right|_{\Sigma},
\left.(p_q f)\right|_{\Sigma}
\right\rangle_{L^2(\Sigma)},
\end{align*}
where 
\[
\alpha_1 = t_q : = \frac{\mu_q+m}{2\mu_q}, \quad 
\alpha_2 = \frac{1-t_q}{2bq} = \frac{1}{2\mu_q(\mu_q+m)}, \quad 
\beta = \frac{1}{\mu_q}.
\]

For $Q_S$, by exploiting that $S$ is compact, we deduce that for all $\delta>0$, there exists a finite-codimension subspace of $L^2(\Sigma)$ such that
\begin{equation*}\label{12nov25b}
|Q_S[f]|\leq \delta ( \| \left.(p_q f)\right|_\Sigma\|_{L^2(\Sigma)}^2 + \| \left.(a p_q f)\right|_\Sigma\|_{L^2(\Sigma)}^2).
\end{equation*}


Using the identities \eqref{a_partial} we have that  
\[
a = -i \left(2\overline{\partial} + \frac{zb}2 \right)= -i e^{-b|z|^2/4} 2\overline{\partial} e^{b|z|^2/4}; \qquad 
a^* = -i \left(2{\partial} - \frac{\overline{z} b}2\right)= -i e^{b|z|^2/4} 2{\partial} e^{-b|z|^2/4}.
\]
Thus 
$\ker L= \ker (a^*a) = \ker (a)$ is given by the functions $u \in L^2(\mathbb R^2)$ such that $f: z \mapsto e^{b|z|^2/4} u(z) $ is entire.
Then, as in previous works on magnetic Hamiltonians, we consider the Fock space $\mathcal F^2$ of entire functions $F$ such that
$$
\|F\|_{\mathcal{F}^2}^2=\int_{\mathbb C}|Fz)|^2 e^{-b|z|^2/2}dm(z)<\infty ,
$$
in such a way that $\ker L = e^{-b|z|^2/4} \mathcal F^2  $ and 
\[
\mathcal H_n:=
\ker (L-\Lambda_n)= \ker (a^*a - 2bn) = (a^*)^n \ker( a) = (a^*)^n \left( e^{-b|z|^2/4} \mathcal F^2  \right) . 
\]
Recall now that we have  the unitary operators 
\[
C_n (a^*)^n: \mathcal H_0 \to  \mathcal H_n, \qquad C_n^{-1} :=\sqrt{(2b)^{n}(n!)},
\]
and that the map
\[
\mathcal{U}:\mathcal{F}^{2} \to \mathcal{H}_0,\quad F\mapsto e^{\frac{-b|\cdot|^2}{4}} F
\]
is also unitary.

Using  the identities
\begin{align*}
(a^{*})^{n} \left( e^{-\frac{b|\cdot|^{2}}{4}} F \right)(z)
&= (-i)^ne^{-\frac{b|z|^{2}}{4}} (2\partial-b\bar z)^{n} F(z)\\
a(a^{*})^{n} p_0 &= 2nb (a^{*})^{n-1} p_0,
\end{align*}
we have that for any $p_n f \in  \mathcal H_n$ there exists $F \in \mathcal{F}^{2}$ such that 
\[
p_n f = C_n e^{\frac{-b|\cdot|^2}{4}} (2\partial-b\bar z)^n  F \quad ; \qquad 
a p_n f = \sqrt{2 n b} \, C_{n-1} e^{\frac{-b|\cdot|^2}{4}} (2\partial-b\bar z)^{n-1}  F .
\]
Then, if on $\mathcal{F}^{2}$  we define the quadratic forms
\[
A_n(v)[F]
= \int_{\Sigma} v(s)
e^{-\frac{bs^{2}}{2}}
\left| (2\partial-b\bar z)^{n} F (s)\right|^2ds,
\]
and
\[
B_n(W_\Sigma)[F]=2{\rm Re} \int_{\Sigma} \left(W_\Sigma\left.\left[e^{-\frac{b|\cdot|^{2}}{4}} (2\partial-b\bar z)^{n} F\right]\right|_{\Sigma}\right)
\overline{ \left.\left[e^{-\frac{b|\cdot|^{2}}{4}} (2\partial-b\bar z)^{n-1} F\right]\right|_{\Sigma} } \mathrm{d} s,
\]
we have:

\begin{lemma}\label{lem55} The quadratic form $Q=Q_{\mathcal A}+Q_S$ is unitarily equivalent to a quadratic form $\widetilde Q= \widetilde Q_{\mathcal A} + \widetilde Q_{S}$ defined on the Fock space $\mathcal F^2$ where 
\[
\widetilde Q_{\mathcal A} := \left\{ \begin{array}{ll}
\widetilde\alpha_1\,A_q(V_1-\delta) + \widetilde\alpha_2\,A_{q-1}(V_2-\delta) 
+ \widetilde\beta\,B_q(W_\Sigma), & \text{ for } q\geq 1,  \\
A_0(V_1-\delta),  & \text{ for } q = 0, 
\end{array} \right.
\]
for
\[
\;\widetilde\alpha_1 = \frac{1}{2\mu_q(\mu_q-m)(2b)^{q-1}(q-1)!}, \quad
\widetilde\alpha_2 = \frac{\mu_q-m}{2\mu_q(2b)^{q-1}(q-1)!}, \quad
\widetilde\beta = \frac{1}{2\mu_q(2b)^{q-1}(q-1)!}\;,
\]
and for all  $ \tilde\delta >0$, there exists a finite-codimension subspace of $\mathcal{F}^{2}$ such that
\[
|\widetilde Q_S | \leq \tilde\delta ( A_q(1) + A_{q-1}(1) ).
\]
\end{lemma}

Let us note that for any $\tilde\delta>0$,
\begin{align*}
|B_n(W_\Sigma)[F]|  \leq& 2  \left(  \int_{\Sigma} \left| W_\Sigma \left[ e^{-\frac{b|\cdot|^2}{4}} (2\partial-b\bar z)^{n-1} F \right]_\Sigma\right|^2 \mathrm{d} s \right)^{1/2}
\left( \int_{\Sigma} e^{-\frac{bs^2}{4}}  \left| (2\partial-b\bar z)^{n} F(s) \right|_\Sigma^2 \mathrm{d} s \right)^{1/2}\\
\leq & 2\| W_\Sigma \| \left|\left| \left[ e^{-\frac{b|\cdot|^2}{4}} (2\partial-b\bar z)^{n-1} F \right]_\Sigma\right|\right|_{L^2(\Sigma)} A_{n}(1)[F]^{1/2} \\
\leq&2\tilde\delta\| W_\Sigma \| A_n(1)[F]+2{\tilde\delta}^{-1} \| W_\Sigma \|A_{n-1}(1)[F].
\end{align*}
Moreover, for $V_1$ positive on the compact set $\Sigma$, there exists  $v_1>0$ such that  $V_1-\tilde\delta  \geq v_1$. 
Then since $V_1$, $V_2$ and $W_\Sigma$  are bounded, by taking 
$\tilde\delta$ small enough, from Lemma \ref{lem55}, we deduce that there exist $c_0>0$, $C_0>0$ and $C_1\in \rr$ such that  on a subspace of $\mathcal F^2$ of finite-codimension:
\begin{equation}\label{>Qt>}
C_0 ( A_q (1) + A_{q-1} (1)) \geq \widetilde Q \geq c_0 A_q (1) - C_1  A_{q-1} (1).
\end{equation}
Now we conclude as in the proof of \cite[Proposition 4.1]{PushRoz07}. By using the compact embedding of $H^1(\Sigma)$ into $L^2(\Sigma)$ and analytic properties, for any $\tilde\delta>0$, there exists a subspace of $\mathcal F^2$ of finite-codimension on which 
\[
\int_\Sigma  \Big| \partial^k F(s) \Big|^2 ds \leq \tilde\delta^2 \int_\Sigma  \Big| \partial^q F(s) \Big|^2 ds, \qquad \forall k=0,1,\cdots, q-1.
\]
We deduce that for any $\tilde\delta>0$ sufficiently small, there exist $c_{\tilde\delta}>0$, $C_{\tilde\delta}>0$ and a subspace of finite-codimension on which 
\[ 
A_{q-1} (1)[F] \leq \tilde\delta \| \partial^q F \|_{L^2(\Sigma)}^2 \quad , \qquad c_{\tilde\delta} \| \partial^q F \|_{L^2(\Sigma)}^2 \leq A_{q} (1)[F] \leq C_{\tilde\delta} \| \partial^q F \|_{L^2(\Sigma)}^2 .
\]
By exploiting these estimates in \eqref{>Qt>} and coming back to the quadratic form $Q$ (defined by \eqref{Def Q}) we obtain that for some positive constants $c_\Sigma$, $C_\Sigma$,
\begin{equation}\label{13nov25}
C_\Sigma \| \left.(p_q f)\right|_\Sigma\|_{L^2(\Sigma)}^2 \geq     Q[f] \geq  c_\Sigma \| \left.(p_q f)\right|_\Sigma\|_{L^2(\Sigma)}^2
\end{equation}
on a subspace of finite codimension in $L^2(\mathbb R^2)$.

Define in $L^2(\mathbb R^2)$  the operator  $p_q\delta_\Sigma p_q$ induced by the quadratic form  
\[
\int_\Sigma  \Big|\left.(p_q f)\right|_\Sigma\Big|^2 ds.
\]
Then from lemma \ref{le1}, the estimates \eqref{13nov25},  and  the min-max principle, it follows
\begin{align*}
	n_+(\lambda , c_\Sigma p_q\delta_\Sigma p_q)+O(1)&\leq \mathcal N_+^q(\lambda)\leq  n_+(\lambda , C_\Sigma p_q\delta_\Sigma p_q)+O(1).
	\end{align*}
Using \cite[Proposition 4.1 (ii)]{PushRoz07} combined with \cite[Corollary 5.11]{BrRa20} we get, for any $C>0$, the eigenvalue asymptotic 
\begin{equation*}
n_+(\lambda; C p_q\delta_\Sigma p_q)=
 \frac{|\ln{\lambda}|}{\ln_2(\lambda)} + \frac{|\ln{\lambda}|\ln_3(\lambda)}{\ln_2(\lambda)^2}  +   \frac{|\ln{\lambda}|}{\ln_2(\lambda)^2}\big(\mathfrak C(\Sigma)+o(1)\big).\end{equation*}
This completes the proof of  theorem \ref{Cluster}.
\end{proof}

\section{Essential self-adjointness and basic spectral properties of $D_{\epsilon,\tau}$: Critical case}\label{S6}
In this section,  we focus on the spectral analysis of $D_{\epsilon,\tau}$ in the critical case. Henceforth, we assume that
\begin{align}\label{Assumption}
\epsilon, \tau \in C^\infty(\Sigma,\rr)\quad \text{with}\quad \epsilon^2(s)-\tau^2(s)=4 \quad \text{for all }\, s\in \Sigma. 
\end{align}
The main change in this setting  is loss of the Fredholm property of  $\Pi_z$, which was a key tool in proving the self-adjointness in the non-critical case.  We recall that $\Pi_z$ and $\underline{\Pi_z}$ are given in \eqref{def_Pi}.

\begin{lemma}\label{Pi_notFredholm} Assume the condition \eqref{Assumption}. Then, for any $z\in\rho(D_0)$ the following holds:
\begin{enumerate}
\item[(i)]  The operator $\Pi_z$ is not Fredholm in $L^2(\Sigma;\cc^2)$. 
\item[(ii)] The operators $\Pi_z\underline{\Pi_z}$ and $\underline{\Pi_z}\Pi_z$ are bounded from $H^{-\frac{1}{2}}(\Sigma;\cc^2)$ to $H^{\frac{1}{2}}(\Sigma;\cc^2)$.
\end{enumerate}
\end{lemma}
\begin{proof} By Proposition \ref{Basic_C_z} there are bounded operators $L_z,\, \widetilde{L_z}: H^{-\frac{1}{2}}(\Sigma;\cc^2)\rightarrow H^{\frac{1}{2}}(\Sigma;\cc^2)$ such that
\begin{align*}
\Pi_z =  \begin{pmatrix} 
	      \frac{\epsilon -\tau}{4} & H_\Sigma\\ 
	H^*_\Sigma & \frac{\epsilon +\tau}{4}  \, 
	   \end{pmatrix}  + L_z, \quad  \underline{\Pi_z} =  \begin{pmatrix} 
	      \frac{\epsilon +\tau}{4} & -H_\Sigma\\ 
	-H^*_\Sigma & \frac{\epsilon-\tau}{4}  \, 
	   \end{pmatrix}  +\widetilde{L_z},
\end{align*}
and the determinant of the principal symbol of the first term in the right hand side of each equality above is
\begin{align*}
\mathrm{det}\begin{pmatrix} 
	      \frac{\epsilon \mp \tau}{4} & \mp\frac{1}{2}T \mathrm{sgn}(\xi)\\ 
	\mp \frac{1}{2}\overline{T} \mathrm{sgn}(\xi) & \frac{\epsilon \pm \tau}{4}  \, 
	   \end{pmatrix} = \frac{\epsilon^2-\tau^2}{16} -\frac{1}{4}=0 \, , \quad \forall \xi\neq0,
\end{align*}
which implies that $\Pi_z$ and  $\underline{\Pi_z}$ are not elliptic, in particular,  not Fredholm in $L^2(\Sigma;\cc^2)$.  Now a simple computation shows
\begin{align*}
 \begin{pmatrix} 
	      \frac{\epsilon -\tau}{4} & H_\Sigma\\ 
	H^*_\Sigma & \frac{\epsilon +\tau}{4}  \, 
	   \end{pmatrix}   \begin{pmatrix} 
	      \frac{\epsilon +\tau}{4} & -H_\Sigma\\ 
	-H^*_\Sigma & \frac{\epsilon-\tau}{4}  \, 
	   \end{pmatrix}= \begin{pmatrix} 
	   \frac{\epsilon +\tau}{4} & -H_\Sigma\\ 
	   -H^*_\Sigma & \frac{\epsilon-\tau}{4}  \, 
	   \end{pmatrix}\begin{pmatrix} 
	   \frac{\epsilon -\tau}{4} & H_\Sigma\\ 
	   H^*_\Sigma & \frac{\epsilon +\tau}{4}  \, 
	   \end{pmatrix}  =  0 \quad \mathrm{ mod }\,\,  \Op \cS^{-1}_2(\Sigma).
\end{align*}
This together with the properties of $L_z$ and $ \widetilde{L_z}$ gives (ii).   
\end{proof}

 Next, we prove the essential self-adjointness of $D_{\epsilon,\tau}$. As will become apparent later in Proposition \ref{BS and Krein 2}, the domain of its closure denoted $\dom \overline{D_{\epsilon,\tau}}$, is not contained in  $H^1(\rr^2\setminus\Sigma;\cc^2)$.

\begin{theorem}\label{main_th2} Let $D_{\epsilon,\tau}$  be as in \eqref{def_A_eta_tau}. Then, under the assumption  \eqref{Assumption},  the operator $(D_{\epsilon,\tau},\dom D_{\epsilon,\tau})$ is essentially self-adjoint and there holds 
\begin{equation*} 
		\begin{split}
			\overline{D_{\epsilon,\tau}}u = (D u_+) \oplus (D u_-), \\
			\dom \overline{D_{\epsilon,\tau}} = \Big\{ u = u_+ \oplus u_- \in  H_\sigma(\Omega_+,\cc^2)&\oplus H_{\sigma, A}(\Omega_-,\cc^2): \\
           &\, \eqref{TC} \text{ holds in } H^{-\frac{1}{2}}(\Sigma;\cc^2) \Big\}.
		\end{split}
	\end{equation*}
\end{theorem}
    
\begin{proof} Following the same arguments as in \cite{BB1} one shows that the adjoint operator $D^{\ast}_{\epsilon,\tau}$ acts as  $D^\ast_{\epsilon,\tau}v = (D_0 v_+) \oplus (D_0 v_-)$ on the domain
\[
\dom D^\ast_{\epsilon,\tau} = \Big\{ v = v_+ \oplus v_- \in  H_\sigma(\Omega_+,\cc^2)\oplus H_{\sigma, A}(\Omega_-,\cc^2): 
           \, \eqref{TC} \text{ holds in } H^{-\frac{1}{2}}(\Sigma;\cc^2) \Big\}.
\]
Since $D_{\epsilon,\tau}$ is symmetric by Lemma \ref{Closable}, it follows that  $D_{\epsilon,\tau}\subset D^{\ast}_{\epsilon,\tau}$. Thus, taking the adjoint yields the inclusion $\overline{D_{\epsilon,\tau}}\equiv D^{\ast\ast}_{\epsilon,\tau}\subset D^{\ast}_{\epsilon,\tau}$. Hence, it remains to prove the inclusion $D^{\ast}_{\epsilon,\tau}\subset\overline{D_{\epsilon,\tau}}$. To do so, given $v\in \dom D^{\ast}_{\epsilon,\tau}$ we will construct a sequence  $(v_j)_{j\in\nn}\subset \dom D_{\epsilon,\tau}$ such that 
 \[
 (v_j, D_{\epsilon,\tau}v_j) \xrightarrow[j\to\infty]{} (v, D^\ast_{\epsilon,\tau}v) \,\,\text{in }\,\, L^2(\rr^2;\cc^2)\times L^2(\rr^2;\cc^2).
 \]
Fix $v\in \dom D^{\ast}_{\epsilon,\tau}$ and set 
\[
g=i\sigma\cdot\nu(\gamma_D^+ v_+-\gamma_D^- v_-)\in H^{-\frac{1}{2}}(\Sigma;\cc^2).
\]
We claim that $\Pi_0 g\in H^{\frac{1}{2}}(\Sigma;\cc^2)$. Indeed, define $u=v-\Phi_0g \in  H_\sigma(\Omega_+,\cc^2)\oplus H_{\sigma, A}(\Omega_-,\cc^2)$, then the jump formula \eqref{eq_C_z_formula} gives
\begin{align}\label{simplify1}
\gamma_D^\pm u_\pm= \gamma_D^\pm v_\pm \pm \frac{i}{2}(\sigma\cdot\nu) g -\mathscr{C}_0 g \equiv \frac{1}{2}(\gamma_D^+ v_+ +\gamma_D^- v_-) - \mathscr{C}_0 g.
\end{align}
Hence, $\gamma_D^+ u_+=\gamma_D^-u_-$ and thus
  \[
    D_0u= (D u_+)\oplus(D u_-) +i\sigma\cdot\nu(\gamma_D^+u_+-\gamma_D^-u_-)\delta_\Sigma=(D_0 u_+)\oplus(D_0 u_-).
    \]
Since $D_0$ is invertible this implies that  $u\in \dom D_0$ and that  $u+\Phi_0g \in \dom D^{\ast}_{\epsilon,\tau}$.
 In particular, we have $\gamma_D^\Sigma u \in H^{\frac{1}{2}}(\Sigma;\cc^2)$. On the other hand, \eqref{simplify1} can be rewritten as 
\begin{align*}
\gamma_D^\Sigma u&= \frac{1}{2}(\gamma_D^+ v_+ +\gamma_D^- v_-) + \frac{1}{4}(\epsilon I_2-\tau\sigma_3) g   -\Pi_0g,\\
&\equiv \frac{1}{4}(\epsilon I_2-\tau\sigma_3)\left[\frac{1}{2} (\epsilon I_2+\tau\sigma_3)(\gamma_D^+ v_+ +\gamma_D^- v_-) + i\sigma\cdot\nu(\gamma_D^+ v_+-\gamma_D^- v_-)\right]   -\Pi_0g.
\end{align*}
Since $v$ satisfies the transmission condition in $H^{-\frac{1}{2}}(\Sigma;\cc^2)$, it follows that $\Pi_0g=- \gamma_D^\Sigma u \in H^{\frac{1}{2}}(\Sigma;\cc^2)$, proving the claim.

Next, let $(f_j)_{j\in\nn}\subset H^{\frac{1}{2}}(\Sigma;\cc^2)$ be such that $f_j \xrightarrow[j\to\infty]{} g$ in $H^{-\frac{1}{2}}(\Sigma;\cc^2)$ and define the sequence of functions
 \begin{align*}
  g_j:= \frac{2}{\epsilon} (\Pi_0g +\underline{\Pi_0} f_j )\equiv g +\frac{2}{\epsilon}  \underline{\Pi_0}[f_j - g],\quad\forall j\in\mathbb{N}.
  \end{align*}
Clearly, $g_j,\, \Pi_0g_j\in H^{\frac{1}{2}}(\Sigma;\cc^2)$ for any $j\in\nn$, and $g_j \xrightarrow[j\to\infty]{} g$ in $H^{-\frac{1}{2}}(\Sigma;\cc^2)$. In addition, we have
\begin{align*}
\Pi_0	g_j= \Pi_0g +\Pi_0\frac{2}{\epsilon}  \underline{\Pi_0}[f_j - g]\equiv \Pi_0g  +\frac{2}{\epsilon}\Pi_0  \underline{\Pi_0}[f_j - g]  +[\Pi_0,\frac{2}{\epsilon}]  \underline{\Pi_0}[f_j - g],
\end{align*}
and as the commutator  $[\Pi_0,\frac{2}{\epsilon}]$ belongs to $\Op \cS^{-1}_2(\Sigma)$, by Lemma \ref{Pi_notFredholm}-(ii) it follows that $\Pi_0 g_j \xrightarrow[j\to\infty]{} \Pi_0 g$ in $H^{\frac{1}{2}}(\Sigma;\cc^2)$.

Now, denote by $\mathcal E$ the extension operator from  $ H^{\frac{1}{2}}(\Sigma;\cc^2)$ to  $ \dom D_{0}$ and define the sequence of functions
\[
v_j= u_j +\Phi_0 g_j, \quad\text{with}\quad u_j= u- \mathcal{E}\left(\Pi_0 \frac{2}{\epsilon} \underline{\Pi_0}[f_j-g]\right) \in \dom D_{0}, \quad \text{for all } j\in\nn.
\]  
By definition we have $v_j\in H^1(\Omega_+,\cc^2)\oplus H^1_{\sigma, A}(\Omega_-,\cc^2)$. Moreover, the same computation as above shows that $v_j$ satisfies the transmission condition in $ H^{\frac{1}{2}}(\Sigma;\cc^2)$. Therefore,  $v_j\in \dom D_{\epsilon,\tau}$ for all $j\in\nn$. Hence, the continuity properties of $\Phi_0$ and $\Pi_0\frac{2}{\epsilon}\underline{\Pi_0}$ implies that 
\begin{align*}
v- v_j&= \mathcal{E}\left(\Pi_0 \frac{2}{\epsilon} \underline{\Pi_0}[f_j-g]\right) - \Phi_0(g- g_j) \xrightarrow[j\to\infty]{}0 \quad \text{in } L^2(\rr^2;\cc^2)\\
D^\ast_{\epsilon,\tau}v- D_{\epsilon,\tau}v_j&= D\mathcal{E}\left(\Pi_0 \frac{2}{\epsilon} \underline{\Pi_0}[f_j-g]\right) \xrightarrow[j\to\infty]{}0 \quad \text{in } L^2(\rr^2;\cc^2).
\end{align*}
Summing up, for any $v\in \dom D^\ast_{\epsilon,\tau}$ we constructed a sequence $(v_j)_{j\in\nn}\subset \dom D_{\epsilon,\tau}$ with $(v_j, D_{\epsilon,\tau}v_j) \to(v, D^\ast_{\epsilon,\tau}v)$  in $L^2(\rr^2;\cc^2)\times L^2(\rr^2;\cc^2)$, for $j\to\infty$. Therefore $ D^{\ast}_{\epsilon,\tau}\subset\overline{D_{\epsilon,\tau}}$ and the theorem is proved. 
\end{proof}
\begin{remark}\label{Remark decomp} A slight change in the proof above actually shows that for any $z\in\rho(D_0)$ and $v\in\dom \overline{D_{\epsilon,\tau}}$ there are unique functions $u\in\dom D_0$ and $g\in H^{-\frac{1}{2}}(\Sigma;\cc^2)$ such that $v= u+ \Phi_z g$ and $\gamma_D^\Sigma u=-\Pi_z g \in H^{\frac{1}{2}}(\Sigma;\cc^2)$.
\end{remark}

We next state the Birman-Schwinger principle and Krein resolvent formula for the critical case.

\begin{proposition}\label{BS and Krein 2}  The following hold:
\begin{itemize}
\item[(i)] For all $z\in\rho(D_0)$ one has $\ker (\overline{D_{\epsilon,\tau}}-z)=\Phi_z\{ g\in H^{-\frac{1}{2}}(\Sigma;\cc^2): g\in\ker \Pi_z\}$. In particular, it holds that $\mathrm{dim}\ker (\overline{D_{\epsilon,\tau}}-z)=\mathrm{dim}\ker \Pi_z$.
\item[(ii)]  For any $z\in\rho(D_0)\cap\rho(\overline{D_{\epsilon,\tau}})$ the operator $\Pi_z$ is bounded bijective from $\{ g\in H^{-\frac{1}{2}}(\Sigma;\cc^2): \Pi_z g\in H^{\frac{1}{2}}(\Sigma;\cc^2)\}$ to $H^{\frac{1}{2}}(\Sigma;\cc^2)$. In particular, $\Pi_z$ admits a bounded inverse from $H^{\frac{1}{2}}(\Sigma;\cc^2)$ to $H^{-\frac{1}{2}}(\Sigma;\cc^2)$, and  one has 
\begin{equation}\label{Krein2}
				(\overline{D_{\epsilon, \tau}} - z)^{-1} = (D_0 - z)^{-1} - \Phi_{z} \Pi_{z}^{-1} \Phi_{\Bar{z}}^*,
\end{equation}	
\item[(iii)] We have  $\dom \overline{D_{\epsilon, \tau}} \not\subset H^1(\rr^2\setminus\Sigma;\cc^2)$.
\end{itemize}
\end{proposition}
\begin{proof} The proof of assertion (i) is similar to that of Lemma \ref{Pi_Fredholm}-(ii) taking into account the fact that the function $g$ in the second identity of \eqref{Traceeigen} is only in $H^{-\frac{1}{2}}(\Sigma;\cc^2)$ in view of the regularity of functions in $\dom \overline{D_{\epsilon,\tau}}$.

 We now prove (ii). Fix $z\in\rho(D_0)\cap\rho(\overline{D_{\epsilon,\tau}})$ and note that  assertion (i) implies that $\ker \Pi_z=\{0\}$. Moreover, by Remark  \ref{Remark decomp}  any $v\in \dom \overline{D_{\epsilon,\tau}}$ can be written as 
 \begin{align*}
 v= (D_0-z)^{-1}f +\Phi_z g
 \end{align*}
  where $f\in L^2(\rr^2;\cc^2)$ and $g\in H^{-\frac{1}{2}}(\Sigma;\cc^2)$ with 
 \begin{align*}
\Pi_z g= - \gamma_D^\Sigma  (D_0-z)^{-1}f\equiv -\Phi_{\Bar z}^\ast f \in H^{\frac{1}{2}}(\Sigma;\cc^2).
 \end{align*}
 Note that $(\overline{D_{\epsilon,\tau}}-z) v= (D-z)(D_0-z)^{-1}f=f$ and thus $v=(\overline{D_{\epsilon,\tau}}-z)^{-1}f$.  Consequently, for any $f\in L^2(\rr^2;\cc^2)$ one has 
  \begin{align}\label{IDD2}
 (\overline{D_{\epsilon,\tau}}-z)^{-1}f= (D_0-z)^{-1}f +\Phi_z g \quad \text{with }\,  \Pi_z g=  -\Phi_{\Bar z}^\ast f.
 \end{align}
 This together with the fact that $\ran  \Phi_{\Bar z}^\ast= H^{\frac{1}{2}}(\Sigma;\cc^2)$ imply that $\ran \Pi_z= H^{\frac{1}{2}}(\Sigma;\cc^2)$. Therefore, the mapping 
 \[
 \Pi_z:  \mathcal G:=\{ g\in H^{-\frac{1}{2}}(\Sigma;\cc^2): \Pi_z g\in H^{\frac{1}{2}}(\Sigma;\cc^2)\} \longrightarrow  H^{\frac{1}{2}}(\Sigma;\cc^2)
 \]
 is well-defined and bijective. Moreover, as $\Pi_z$ is continuous and injective in $H^{-\frac{1}{2}}(\Sigma;\cc^2)$ one easily deduce that $\Pi_z:  \mathcal G \longrightarrow  H^{\frac{1}{2}}(\Sigma;\cc^2)$ is closed. Hence, $\Pi_z^{-1}: H^{\frac{1}{2}}(\Sigma;\cc^2) \longrightarrow \mathcal G $ is everywhere defined and closed, and therefore bounded by the closed graph theorem. Thus, $\Pi_z$ admits a bounded inverse from $H^{\frac{1}{2}}(\Sigma;\cc^2)$ to $H^{-\frac{1}{2}}(\Sigma;\cc^2)$ and plugging $g= - \Pi_z^{-1}\Phi_{\Bar z}^\ast f$ in \eqref{IDD2} yields the resolvent formula \eqref{Krein2}.
 
 Let us show (iii). By contradiction, suppose that $\dom \overline{D_{\epsilon, \tau}} \subset H^1(\rr^2\setminus\Sigma;\cc^2)$. Then, for any  $z\in\cc\setminus\rr$ the resolvent formula implies that $\Phi_{z} \Pi_{z}^{-1} \Phi_{\Bar{z}}^*f \in H^1(\rr^2\setminus\Sigma;\cc^2)$ for all $f \in L^2(\rr^2;\cc^2)$. Since $\ran  \Phi_{\Bar z}^\ast= H^{\frac{1}{2}}(\Sigma;\cc^2)$, applying the jump formula \eqref{eq_C_z_formula} yields that $\ran  \Pi_{z}^{-1}\subset H^{\frac{1}{2}}(\Sigma;\cc^2)$. This shows that $\Pi_{z}$ is boundedly invertible in $H^{\frac{1}{2}}(\Sigma;\cc^2)$. However, from the proof of Lemma \ref{Pi_notFredholm} we have seen that $\Pi_{z}$  is a non elliptic zero-order pseudodifferential operator, in particular, $\Pi_{z}$ is not invertible in any $H^s(\Sigma;\cc^2)$ for any $s\in\rr$ leading to a contradiction. Therefore, $\dom \overline{D_{\epsilon, \tau}} \not\subset H^1(\rr^2\setminus\Sigma;\cc^2)$.
\end{proof}
\begin{remark} In the proof of (iii) we only used that $\Pi_{z}$ is not invertible in $H^{\frac12}(\Sigma;\cc^2)$. By exploiting more carefully that the non elliptic part of the operator $\Pi_{z}$ is a pseudo-differential of order $-1$ (because $0$ is an eigenvalue of its $0$-order principal symbol; see the proof of Lemma \ref{Pi_notFredholm})  then (iii) could be extended to any $H^s(\rr^2\setminus\Sigma;\cc^2)$, $s>0$ (see also Remark \ref{extTr}). 
\end{remark}

We next give a complete characterization of the essential spectrum of $\overline{D_{\epsilon,\tau}}$.
\begin{theorem}\label{Ess Spec}Under the assumption \eqref{Assumption} we have $\spece \overline{D_{\epsilon,\tau}}= \spec D_0 \cup I_\Sigma$, where
\begin{align}\label{Def I Sigma}
I_\Sigma:=  \ran\left(-m\frac{\tau}{\epsilon}\right).
\end{align}
\end{theorem}
\begin{remark} Note that, in view of the condition \eqref{Assumption},  $I_\Sigma=-m \frac{\tau}{\epsilon}(\Sigma)$ is a closed interval satisfying $I_\Sigma \subset (-m, m)$ and reduces to a single point if and only if $\epsilon$ and $\tau$ are constants or $m=0$.
\end{remark}

The remainder of this section will be devoted to the proof of Theorem \ref{Ess Spec}. As first step we have:
\begin{lemma}\label{Ess Spec1} It holds that $\spec D_0\subset \spece \overline{D_{\epsilon,\tau}}$.
\end{lemma}
\begin{proof} The result will follow by constructing a singular sequence for each Landau--Dirac level. 

For  $c=(c_1,c_2)\in \rr^2$ define the map $U_c: L^2(\rr^2;\cc^2) \longrightarrow L^2(\rr^2;\cc^2)$ by
\[
U_c f (x):= e^{i\frac{b}{2}(c_1 x_2-c_2x_1)}f(x-c).
\]
Clearly, $U_c$ is a unitary operator on $L^2(\rr^2;\cc^2)$ and maps isomorphically $H^1(\rr^2;\cc^2)$ into itself.  Moreover, it is straightforward to check that $D_0 U_c=U_cD_0$.

Let $\mu \in\spec D_0$ be a Landau--Dirac level and let $\phi$ be an associated normalized eigenfunction, i.e., $D_0 \phi=\mu\phi$ and $\| \phi\|=1$. Note that, in view of the resolvent kernel of $D_0$ (see Lemma \ref{prop:explicit_varphi}), there exist $N\in \nn$ and $C_N>0$ such that the following Gaussian bound holds
\begin{align*}
|\phi(x)| \leq C_N(1+|x|)^Ne^{-\frac{b}{4}|x|^2 }\quad \forall x\in\rr^2,
\end{align*}
see also \cite[Sec. 3.1]{RaWa02} for more details.
Then, for $u_c= U_c\phi$ we have $\| u_c\|=1$, $D_0u_c=U_c D_0\phi=\mu U_c\phi=\mu u_c$, and 
\begin{align}\label{Exp decay}
|u_c(x)| \leq C_N(1+|x-c|)^Ne^{-\frac{b}{4}|x-c|^2} \quad \forall x\in\rr^2.
\end{align}
Fix $R>0$ with $\Sigma\subset B(0,R)$. Pick a sequence of points $\{x_k\}_{k\in\nn}\subset \rr^2$ with $|x_k| \rightarrow \infty$, and set $R_k:= |x_k|/3$. Let $\chi\in C_0^\infty([0,\infty);[0,1])$ be a cut-off function with $\chi\equiv 1$ on $[0,1]$, $\chi\equiv 0$ on $[2,\infty)$, and set
\[
\chi_k(x)= \chi(|x-x_k|/R_k),\quad u_k= U_{x_k}\phi, \quad v_k=\chi_k u_k.
\] 
For $k$ large enough (so that $R_k>R$), we have 
\[
\supp \chi_k\subset \overline{B(x_k,2R_k)} \quad\text{and}\quad \mathrm{dist}\,( \overline{B(x_k,2R_k)}, \Sigma)\geq R_k-R  \xrightarrow[k\to\infty]{} \infty,
\]
so  $\supp v_k \subset \Omega_-$ which implies that $v_k$  vanishes identically in a neighbourhood of $\Sigma$. Remark that $|1-\chi_k|\leq 1$ and is supported in $|x-x_k|\geq R_k$. From this and the Gaussian bound \eqref{Exp decay}, we conclude by the dominated convergence theorem that 
\[
| \| v_k\|_{L^2(\rr^2;\cc^2)}-1|\leq  \| (\chi_k^2-1)u_k\|_{L^2(\rr^2;\cc^2)}\leq \|u_k\|_{L^2(\rr^2\setminus B (x_k, R_k);\cc^2)}  \xrightarrow[k\to\infty]{} 0,
\]
and thus $\| v_k\| \xrightarrow[k\to\infty]{} 1$. Moreover, for any $\varphi\in C_0^\infty(\rr^2;\cc^2)$ with $\supp \varphi\subset B(0,\tilde{R})$ for some $\tilde R>0$, we have 
\begin{align*}
|\langle v_k,\varphi\rangle|\leq \int_{ B(0,\tilde{R})} |u_k(x)| |\varphi(x)|\mathrm{d}x &\leq \|\varphi\| _{L^2(\rr^2;\cc^2)} \|u_k\| _{L^2(B(0,\tilde{R});\cc^2)}\\
&\equiv \|\varphi\| _{L^2(\rr^2;\cc^2)} \|\phi\| _{L^2(B(x_k,\tilde{R});\cc^2)} \xrightarrow[k\to\infty]{} 0,
\end{align*}
because $|x_k|  \xrightarrow[k\to\infty]{} \infty$ and $\tilde R$ is fixed. Since $C_0^\infty(\rr^2;\cc^2)$ is dense in $L^2(\rr^2;\cc^2)$, from the above it follows that $v_k$ converges weakly to $0$.

We proceed to show that $\| (D_{\epsilon,\tau}-\mu)v_k\|_{L^2(\rr^2;\cc^2)} \xrightarrow[k\to\infty]{} 0$. We have 
\[
(D_{\epsilon,\tau}-\mu)v_k= (D-\mu)v_k= \chi_k(D-\mu)u_k -(i\sigma\cdot\nabla \chi_k)u_k\equiv  -(i\sigma\cdot\nabla \chi_k)u_k.
\]
Observe that $\| \nabla\chi_k\|_{\infty}= \| \nabla\chi\|_{\infty}/ R_k$, $\supp\nabla\chi_k\subset \{ R_k\leq |x-x_k|\leq 2 R_x\}$, and 
\begin{align*}
|u_k(x)| \leq C_N(1+2R_k)^Ne^{-\frac{b}{4} R_k^2} \quad \forall x\in \{ R_k\leq |x-x_k|\leq 2 R_x\},
\end{align*}
holds by \eqref{Exp decay}. Hence, using $| B(x_k, 2R_k)|= 4\pi R_k^2$,  we get 
\begin{align*}
\| (D_{\epsilon,\tau}-\mu)v_k\|_{L^2(\rr^2;\cc^2)} &=   \|(i\sigma\cdot\nabla \chi_k)u_k\|_{L^2(\rr^2;\cc^2)}= \|(\nabla \chi_k)u_k\|_{L^2(\rr^2;\cc^2)}\\
&\leq \frac{\|\nabla\chi\|_{\infty}}{R_k} C_N(1+2R_k)^Ne^{-\frac{b}{4} R_k^2} (4\pi R_k^2) \xrightarrow[k\to\infty]{} 0.
\end{align*}
Thus
\begin{align*}
\frac{\| (D_{\epsilon,\tau}-\mu)v_k\|_{L^2(\rr^2;\cc^2)} }{\| v_k\|_{L^2(\rr^2;\cc^2)}}   \xrightarrow[k\to\infty]{} 0.
\end{align*}
Therefore, for $k_0$ large enough,  $(v_k)_{k\geq k_0}\subset \dom D_{\epsilon,\tau}$ is a singular sequence for $\mu$, proving the inclusion $\spec D_0\subset \spece \overline{D_{\epsilon,\tau}}$.
\end{proof}

Let $\Lambda= (S(-1))^{-\frac{1}{2}}$ where $S(-1)$ is single layer boundary operator given in \eqref{SL} with $\mu=-1$. By Lemma  \ref{H properties}, the operator $\Lambda$ is self-adjoint in $L^2(\Sigma)$ and belongs to  $\Op \cS^{\frac{1}{2}}_1(\Sigma)$ with
\[
p_\Lambda(x,\xi)=\sqrt{2|\xi|}\quad \forall \xi\in T^\ast\Sigma.
\]
Moreover,  $\Lambda: H ^{s}(\Sigma)\rightarrow H ^{s-\frac{1}{2}}(\Sigma)$ is an isomorphism for any $s\in\rr$.

For $z\in\rho(D_0)$ we introduce the operator 
\begin{align*}
\mathcal L_z= (\Lambda \otimes I_2)\Pi_z ( \Lambda \otimes I_2)
\end{align*}
acting on the maximal domain 
\[
\dom \mathcal L_z= \{ g\in L^2(\Sigma;\cc^2): \mathcal L_z g\in L^2(\Sigma;\cc^2)\}.
\]

\begin{lemma}\label{L BS} The following holds: 
\begin{itemize}
\item[(i)] For all $z\in\rho(D_0)$ one has $\ker (\overline{D_{\epsilon,\tau}}-z)=\Phi_z\Lambda\{ g\in L^2(\Sigma;\cc^2): g\in\ker \mathcal L_z\}$. In particular, it holds that $\mathrm{dim}\ker (\overline{D_{\epsilon,\tau}}-z)=\mathrm{dim}\ker \mathcal L_z$.
\item[(ii)] For all $z\in\rho(D_0)\cap\rho(\overline{D_{\epsilon,\tau}})$ the operator $\mathcal L_z $ is bounded bijective from  $\dom \mathcal L_z$ to $L^2(\Sigma;\cc^2)$,  and it holds that 
\begin{equation}\label{Krein3}
				(\overline{D_{\epsilon, \tau}} - z)^{-1} = (D_0 - z)^{-1} - \Phi_{z} \Lambda\mathcal L_{z}^{-1} \Lambda \Phi_{\Bar{z}}^*,
\end{equation}
\item[(iii)] For any $z\in\rho(D_0)\cap\rr$ the operator $(\mathcal L_z ,\dom \mathcal L_z)$ is self-adjoint.	
\end{itemize}
\end{lemma}
\begin{proof} Assertions (i) and (ii) are  direct consequences of Proposition \ref{BS and Krein 2}. To prove (iii),  note that for any $z\in\rho(D_0)\cap\rr$,  $\mathcal L_z$ is a symmetric first order pseudodifferential operator and therefore self-adjoint on the maximal domain $\dom \mathcal L_z$ by the Friedrichs lemma, see e.g. \cite[Proposition.~7.4]{T2}.
\end{proof}

\begin{lemma}\label{L relate D} If $z\in \rho(D_0)\cap\rr$ is such that $0\in\spece \mathcal L_z$, then $z\in\spece  \overline{D_{\epsilon,\tau}}$.
\end{lemma}
\begin{proof} Fix $z\in \rho(D_0)\cap\rr$ with $0\in\spece \mathcal L_z$. Then, there exists a sequence $(g_j)_{j\in\nn}\subset \dom \mathcal L_z$ weakly converging to $0$ with $\|g_j\|_{L^2(\Sigma;\cc^2)}=1$ for all $j\in\nn$, and $\| \mathcal L_z g_j\|_{L^2(\Sigma;\cc^2)} \rightarrow 0$ as $j \rightarrow \infty$. 
Recall the extension operator $\mathcal E: H^{\frac{1}{2}}(\Sigma;\cc^2) \rightarrow \dom D_{0}$ and set 
\[
v_j= \mathcal{E}(-\Lambda^{-1}\mathcal L_z g_j)+ \Phi_z\Lambda g_j\quad \forall j\in\nn.
\]
Obviously, $v_j\in H_\sigma(\Omega_+,\cc^2)\oplus H_{\sigma, A}(\Omega_-,\cc^2)$, and by Proposition \ref{Further_Phi_z}  and the definition of $\mathcal L_z$ we have 
\begin{align*}
\gamma_D^\pm v_{j,\pm}= -\Lambda^{-1}\mathcal L_z g_j \mp\frac{i}{2}(\sigma\cdot\nu)\Lambda g_j + \mathscr C_z\Lambda g_j=  -(\frac{1}{4}(\epsilon I_2 - \tau\sigma_3) \pm\frac{i}{2}(\sigma\cdot\nu) ) \Lambda g_j.
\end{align*}
From this we easily check that 
\begin{align*}
\frac{1}{2}(\epsilon I_2+\tau\sigma_3)(\gamma_D^+ v_{j,+} +\gamma_D^- v_{j,-})+ &i\sigma\cdot\nu(\gamma_D^+ v_{j,+} -\gamma_D^- v_{j,-})= 0,
\end{align*}
and thus $v_j\in\dom  \overline{D_{\epsilon,\tau}}$. The task now is to show that there exists $j_0\geq 1$ large enough so that $(v_j)_{j\geq j_0}$ is a singular sequence for $z$. 

Let $\mathcal E': H^{-1}(\rr^2;\cc^2)\rightarrow H^{-\frac{1}{2}}(\Sigma;\cc^2)$ be the dual operator of $\mathcal E$. Then, for $f\in L^2(\rr^2;\cc^2)$ we compute 
\begin{align*}
\langle v_j,f \rangle_{ L^2(\rr^2;\cc^2)} &= \langle \mathcal{E}(-\Lambda^{-1}\mathcal L_z g_j), f\rangle_{ L^2(\rr^2;\cc^2)} + \langle  \Phi_z\Lambda g_j, f \rangle _{ L^2(\rr^2;\cc^2)}\\
&= \langle -\mathcal L_z g_j, \Lambda^{-1}\mathcal{E}'f\rangle_{ L^2(\Sigma;\cc^2)} + \langle   g_j,\Lambda \Phi_{\overline{z}}^\ast f \rangle_{ L^2(\Sigma;\cc^2)}.
\end{align*}
Here we used that  $\Lambda: H^s(\Sigma;\cc^2)\rightarrow H^{s-\frac{1}{2}}(\Sigma;\cc^2)$ is an isomorphism for any $s\in\rr$. Hence
\begin{align*}
\left| \langle v_j,f \rangle_{ L^2(\rr^2;\cc^2)} \right|&\leq   \| \mathcal L_z g_j\|_{ L^2(\Sigma;\cc^2)} \| \Lambda^{-1}\mathcal{E}'f\|_{ L^2(\Sigma;\cc^2)} + \langle   g_j,\Lambda \Phi_{\overline{z}}^\ast f \rangle_{ L^2(\Sigma;\cc^2)}\xrightarrow[j\to\infty]{} 0
\end{align*}
because $(g_j)_{j\in\nn}$ is a Weyl sequence, and so $(v_j)_{j\in\nn}$ is weakly converging to $0$. By Proposition \ref{Basic_Phi}-(i) there exists $c>0$ such that 
\begin{align*}
	\|\Phi_z\Lambda g_j \|_{L^2(\rr^2;\cc^2)}\equiv \|\Phi_z\Lambda g_j \|_{H_{\sigma, A}(\rr^2\setminus\Sigma;\cc^2)} \geq c\|\Lambda g_j\|_{H^{-\frac{1}{2}}(\Sigma;\cc^2)}=c  \| g_j\|_{L^2(\Sigma;\cc^2)}=c,
\end{align*}
and thus 
\begin{align*}
c \leq \|\Phi_z\Lambda g_j\|_{L^2(\rr^2;\cc^2)}&\leq  \| v_j\|_{L^2(\rr^2;\cc^2)} + \| \mathcal{E}(-\Lambda^{-1}\mathcal L_z g_j)\|_{L^2(\rr^2;\cc^2)}\\
&\leq  \| v_j\|_{L^2(\rr^2;\cc^2)} + \| \Lambda^{-1}\mathcal L_z g_j\|_{H^{\frac{1}{2}}(\Sigma;\cc^2)}\\
&=  \| v_j\|_{L^2(\rr^2;\cc^2)} + \| \mathcal L_z g_j\|_{L^2(\Sigma;\cc^2)}.
\end{align*}
This entails that for any $\eta>0$ there exists $j_0\geq 1$ such that for all $j\geq j_0$  we have $\| v_j\|_{L^2(\rr^2;\cc^2)}\geq c-\eta$. 
Next we estimate
\begin{align*}
\| ( \overline{D_{\epsilon,\tau}}- z)v_j\|_{L^2(\rr^2;\cc^2)}&= \| (D-z)\mathcal{E}(-\Lambda^{-1}\mathcal L_z g_j) \|_{L^2(\rr^2;\cc^2)}\\
& \leq \| \Lambda^{-1}\mathcal L_z g_j\|_{H^{\frac{1}{2}}(\Sigma;\cc^2)}= \| \mathcal L_z g_j\|_{L^2(\Sigma;\cc^2)}\xrightarrow[j\to\infty]{} 0.
\end{align*}
Thus, for $j_0\geq 1$ sufficiently large and $j\geq j_0$,
\begin{align*}
\frac{\| ( \overline{D_{\epsilon,\tau}}- z)v_j\|_{L^2(\rr^2;\cc^2)}}{\| v_j\|_{L^2(\rr^2;\cc^2)} }\leq \frac{1}{c-\eta} \| ( \overline{D_{\epsilon,\tau}}- z)v_j\|_{L^2(\rr^2;\cc^2)} \xrightarrow[j\to\infty]{} 0.
\end{align*}
Hence, $(v_j)_{j\geq j_0}$ is a singular sequence for $z$. Therefore $z\in\spece  \overline{D_{\epsilon,\tau}}$ and this completes the proof.
\end{proof}

We next analyze the essential spectrum of $\mathcal L_z$.
\begin{proposition}\label{Spec ess L} Let $I_\Sigma$ be as in \eqref{Def I Sigma}. Then, for any  $z\in  \rho(D_0)$  one has $0\in \spece \mathcal L_z$ if and only if $z\in I_\Sigma$.
\end{proposition}
Before going through the proof of Proposition \ref{Spec ess L} we will need further properties of the operator $H_\Sigma$ defined by \eqref{Cauchy2}. The proof of the following lemma is inspired by \cite[Sec. 5]{bell}.
\begin{lemma}\label{Cummutator smooth} Let $f\in C^\infty(\Sigma,\rr)$ and $H_\Sigma$ be as in \eqref{Cauchy2}.  Then it  holds that 
\begin{align*}
[H_\Sigma,f],\,\,H_\Sigma^\ast H_\Sigma-\frac{1}{4},\, H_\Sigma H^\ast_\Sigma-\frac{1}{4}\,\in  \Op \cS^{-\infty}_1(\Sigma).
\end{align*}
\end{lemma}	
\begin{proof}  
 We parametrize $\Sigma$  by arc length $\rho: \mathbb T \longrightarrow \Sigma$, $\mathbb T = \rr / l\zz$, with  $l$ being the length of $\Sigma$, $\rho(s)=(\rho_1(s),\rho_2(s))$, $|\rho'(s)|=1$. We shall identify $\rr^2$ with $\cc$ and set $z(s)=\rho_1(s)+i\rho_2(s)$. Since $\rho$ is simple and closed, $z: \mathbb T \longrightarrow \cc$ is injective and $z'(s)=\rho'_1(s)+i\rho'_2(s)\neq 0$ for all $s$.

 Now, as the multiplication by $f$ gives rise to a zero order pseudodifferential operator it holds that $ [H_\Sigma,f]\in  \Op \cS^{-1}_1(\Sigma)$. Thus, for all $g\in C^\infty(\Sigma)$ one has 
\[
 [H_\Sigma,f]g(\rho(s))=  \int_{\mathbb T} K(s,t)\,g(\rho(t))\,\mathrm{d} t
\]
with
\[
K(s,t) := \frac{i}{2\pi}\,\frac{f(\rho(s))-f(\rho(t))}{z(s)-z(t)}.
\]
Let $h= f\circ\rho \in C^\infty(\mathbb T)$, then by Taylor's theorem with
integral remainder there holds
\begin{align*}
h(s)-h(t)= (s-t)F(s,t), \quad F(s,t):= \int_0^1 h'(t+ r(s-t))\,\mathrm{d} r,\\
z(s)-z(t)= (s-t)G(s,t), \quad G(s,t):= \int_0^1 z'(t+ r(s-t))\,\mathrm{d}r,
\end{align*}
and both formulas extend continuously to $s=t$ with $F(s,s)=h'(s)$ and $G(s,s)=z'(s)$. Clearly we have $F, G\in C^\infty(\mathbb T\times \mathbb T)$ since $h, z\in C^\infty(\mathbb T)$. Moreover, as $z'(s)\neq 0$ holds for any $s\in \mathbb T$, it follows that $G(s,t)\neq 0$ for all $(s,t)\in \mathbb T\times \mathbb T$. Thus 
\[
K := \frac{i}{2\pi}\, \frac{F}{G}\in C^\infty(\mathbb T\times \mathbb T),
\]
which implies that $[H_\Sigma,f]$ is a smoothing operator, and therefore  $[H_\Sigma,f]\in  \Op \cS^{-\infty}_1(\Sigma)$.

Next, recall that the operator of multiplication $T := t_1+ it_2$ given in \eqref{Def T} where $(t_1,t_2)$ is the tangent vector on $\Sigma$, and that $\mathcal H=-H_\Sigma \overline{T}$ is the Hilbert transform on $\Sigma$. Then, the same reasoning as above shows that the Kerzman--Stein operator  $\mathcal H +\mathcal H^\ast  =-(H_\Sigma \overline{T}+ TH_\Sigma^\ast)$ has a smooth integral kernel, see also \cite[Sec. 5, p. 18]{bell} for more details. Hence, $H_\Sigma \overline{T}+ TH_\Sigma^\ast \in \Op \cS^{-\infty}_1(\Sigma)$, and by Lemma \ref{H properties} we obtain
\begin{align*}
H_\Sigma H_\Sigma^\ast -\frac{1}{4} =H_\Sigma  \overline{T} T H_\Sigma^\ast +(H_\Sigma  \overline{T} )^2= H_\Sigma  \overline{T} (H_\Sigma  \overline{T}+T H_\Sigma^\ast) \in  \Op \cS^{-\infty}_1(\Sigma),
\end{align*}
and taking the adjoint gives $H_\Sigma^\ast H_\Sigma-\frac{1}{4} \in  \Op \cS^{-\infty}_1(\Sigma)$.
\end{proof}

\begin{proof}[Proof of Proposition \ref{Spec ess L}] We only focus on the case $z\in \rr\cap \rho(D_0)$ as $\mathcal L_z$ is invertible for $z\in \rho(D_0)\setminus \rr$. 

Since $\Lambda S(0) \Lambda= I \,\, \mathrm{mod }  \Op \cS^{-1}_1(\Sigma)$, by Proposition \ref{Basic_C_z}  we have 
\begin{align}\label{L decomposition}
\begin{split}
\mathcal L_z&= (\Lambda \otimes I_2)( \frac{1}{4}(\epsilon I_2-\tau\sigma_3)+ \Theta )(\Lambda \otimes I_2) +(m\,\sigma_3 + z\,I_2) +K_z\\
&=\mathcal M + z I +K_z,
\end{split}
\end{align}
where $K_z$  are compact operators in $L^2(\Sigma;\cc^2)$ depending holomorphically on $z\in\rho(D_0)$ due to the holomorphic properties of $\mathscr C_z$ (see Proposition \ref{Further_Phi_z}), and 
\begin{align*}
\begin{split}
\mathcal M :=\begin{pmatrix}
\Lambda\frac{1}{\epsilon+\tau}\Lambda +m &\Lambda H_\Sigma\Lambda  \\
\Lambda H_\Sigma^\ast \Lambda  & \Lambda \frac{1}{\epsilon-\tau}\Lambda -m
\end{pmatrix} 
\end{split}
\end{align*}
Let $\mathbb P$ be the spectral projector onto the finite dimensional subspace $\ker (\Lambda\frac{1}{\epsilon+\tau}\Lambda +(z+m) )$ in $L^2(\Sigma)$. Then $\mathbb P$ is a smoothing operator and one has
\begin{align*}
		0\in \spece \mathcal L_z \Leftrightarrow 0\in \spece (\mathcal M+z) \Leftrightarrow
		0\in\spece \begin{pmatrix}
\Lambda\frac{1}{\epsilon+\tau}\Lambda +(z+m) + \mathbb P &\Lambda H_\Sigma\Lambda  \\
\Lambda H_\Sigma^\ast \Lambda  & \Lambda \frac{1}{\epsilon-\tau}\Lambda +(z-m)
\end{pmatrix}.
		\end{align*}
Since $\mathcal L_z$ is self-adjoint, by the theory of block operator matrices \cite[Thm.~2.4.6]{tretter} it follows that 
\begin{align*}
			0\in \spece \mathcal L_z \Leftrightarrow  0\in \spece \mathcal S_z,
\end{align*}
where, as an operator in $L^2(\Sigma)$,  $\mathcal S_z$ is the Schur complement defined on its maximal domain by
\[
\mathcal S_z= \Lambda \frac{1}{\epsilon-\tau}\Lambda +(z-m) - \Lambda H_\Sigma^\ast \Lambda \left(\Lambda\frac{1}{\epsilon+\tau}\Lambda +(z+m) + \mathbb P \right)^{-1} \Lambda H_\Sigma \Lambda.
\]
We claim that $\mathcal S_z$ is bounded in $L^2(\Sigma)$ and one has  $0\in \spece \mathcal S_z$ if and only if $-z\in m \frac{\tau}{\epsilon}(\Sigma)$. Indeed, observe that 
\[
\mathcal S_z= \Lambda \frac{1}{\epsilon-\tau}\Lambda +(z-m) - \Lambda H_\Sigma^\ast(\epsilon+\tau) \left(1+\Lambda^{-1}(z+m + \mathbb P)\Lambda^{-1} (\epsilon+\tau)\right)^{-1} H_\Sigma \Lambda.
\]
Using the resolvent formulas $(I+C)^{-1}= I-C+ C^2(I+C)^{-1}$ and the fact that $\mathbb P$ is a smoothing operator, we get 
\begin{align*}
\left(1+\Lambda^{-1}(z+m + \mathbb P)\Lambda^{-1} (\epsilon+\tau)\right)^{-1}&= I-(z+m) \Lambda^{-2}(\epsilon+\tau)\quad \mathrm{ mod }\,\,  \Op \cS^{-2}_1(\Sigma)\\
&= I-(z+m) (\epsilon+\tau)\Lambda^{-2}\quad \mathrm{ mod }\,\,  \Op \cS^{-2}_1(\Sigma),
\end{align*}
where in the last equality we used that $[ (\epsilon+\tau),\Lambda^{-2} ]\in \Op \cS^{-2}_1(\Sigma)$. Hence, we can rewrite $\mathcal S_z$ as
\begin{align*}
\mathcal S_z= (z-m) +\Lambda \left(\frac{\epsilon+\tau}{4} -    H_\Sigma^\ast(\epsilon+\tau) H_\Sigma \right)\Lambda +(z+m) \Lambda H_\Sigma^\ast(\epsilon+\tau)^2\Lambda^{-2} H_\Sigma \Lambda \,\, \mathrm{ mod }\,\,  \Op \cS^{-1}_1(\Sigma).
\end{align*}
 Applying Lemma \ref{Cummutator smooth} with $f=\epsilon +\tau$, we obtain
\begin{align*}
\Lambda \left(\frac{\epsilon+\tau}{4} -   H_\Sigma^\ast(\epsilon+\tau) H_\Sigma \right)\Lambda &= \Lambda \left(H_\Sigma^\ast H_\Sigma(\epsilon+\tau) -    H_\Sigma^\ast(\epsilon+\tau) H_\Sigma \right) \Lambda\,\, \mathrm{ mod }\,\,  \Op \cS^{-\infty}_1(\Sigma)\\
&\equiv \Lambda H_\Sigma^\ast [ H_\Sigma,\,\epsilon+\tau] \Lambda \,\, \mathrm{ mod }\,\,  \Op \cS^{-\infty}_1(\Sigma),
\end{align*}
and thus
\begin{align*}
\Lambda \left(\frac{\epsilon+\tau}{4} -   H_\Sigma^\ast(\epsilon+\tau) H_\Sigma \right)\Lambda \in  \Op \cS^{-\infty}_1(\Sigma).
\end{align*}
Using again that $H_\Sigma^\ast H_\Sigma= \frac{1}{4} \,\, \mathrm{ mod }\,\,  \Op \cS^{-\infty}_1(\Sigma)$ and  iterating the commutator arguments as before yield
\begin{align*}
 \Lambda H_\Sigma^\ast(\epsilon+\tau)^2\Lambda^{-2} H_\Sigma \Lambda = \frac{(\epsilon+\tau)^2}{4} \,\, \mathrm{ mod }\,\,  \Op \cS^{-1}_1(\Sigma).
\end{align*}
Consequently,  $\mathcal S_z$ is bounded in $L^2(\Sigma)$ and 
\begin{align*}
\mathcal S_z= (z-m) +\frac{(z+m)(\epsilon+\tau)^2}{4}   \,\, \mathrm{ mod }\,\,  \Op \cS^{-1}_1(\Sigma).
\end{align*}
Notice that 
\begin{align*}
 (z-m) +\frac{(z+m)(\epsilon+\tau)^2}{4} = 2\frac{ z\epsilon +m\tau}{\epsilon-\tau}.
\end{align*}
Therefore, $0\in \spece \mathcal S_z$ if and only if $z\epsilon(s) +m\tau(s)=0$ for some $s\in\Sigma$, and hence 
\begin{align*}
			0\in \spece \mathcal L_z \Leftrightarrow  0\in \spece \mathcal S_z \Leftrightarrow z\in I_\Sigma,
\end{align*}
and this completes the proof.
\end{proof}

We can now provide the  proof of Theorem \ref{Ess Spec} 

\begin{proof}[Proof of Theorem \ref{Ess Spec}] By Lemmas \ref{Ess Spec1} and \ref{L relate D} we have $\spec D_0 \cup I_\Sigma \subset \spece \overline{D_{\epsilon,\tau}}$. So it remains to show the inclusion $\spece \overline{D_{\epsilon,\tau}}\subset \spec D_0 \cup I_\Sigma$. For this, it suffices to show the map $z \mapsto (\overline{D_{\epsilon,\tau}} -z)^{-1}$ extends meromorphically into  $\cc\setminus ( \spec D_0 \cup I_\Sigma)$ with finite-dimensional coefficients in the principal parts of the Laurent series at the poles.

From \eqref{L decomposition} we have $\mathcal L_z=\mathcal M + z I +K_z$ with $\rho(D_0)\ni z\mapsto K_z \in \mathcal L(L^2(\Sigma;\cc^2))$ a holomorphic operator-valued functions, $\mathcal M$ is self-adjoint on $\dom \mathcal L_z$, and  
\[
0\in \spece \mathcal L_z \Leftrightarrow 0\in \spece (\mathcal M+z), \quad \spece \mathcal M = -I_\Sigma,
\]
holds by Proposition \ref{Spec ess L}. Hence,  $z\mapsto (\mathcal M+z)^{-1}$ is meromorphic in $\cc\setminus I_\Sigma$ with simple poles of finite rank residues.  Hence, the operator-valued function $ K_z(\mathcal M + z )^{-1}$ are compact operators depending holomorphically on $\cc\setminus(\spec D_0\cup\spec(-\mathcal M))$ and extend meromorphically to $\cc\setminus(\spec D_0\cup I_\Sigma)$. Therefore, the operator function  
	\[
		z\mapsto (I+{K}_{z}(\mathcal{M}+zI)^{-1}),
	\]
fulfils the assumption of the meromorphic Fredholm theorem on $\cc\setminus(\spec D_0\cup I_\Sigma)$; see, e.g. \cite[Theorem XIII.13]{RS}. 

Now, for any $z\in\cc\setminus(\spec D_0 \cup\spec(-\mathcal M) ) \subset \rho(D_0)$  we write 
\[
\mathcal L_z= (I +K_z(\mathcal M + z )^{-1})(\mathcal M + z ).
\]
 Since $\mathcal L_z$ is boundedly invertible for $z\in\cc\setminus\rr$ (Lemma \ref{L BS}), $(I +K_z(\mathcal M + z )^{-1})$ have the same property, and
 \[
\mathcal L_z^{-1}= (\mathcal M + z )^{-1}(I +K_z(\mathcal M + z )^{-1})^{-1} .
\]
By the meromorphic Fredholm theorem  \cite[Theorem XIII.13]{RS}) we deduce that $z\mapsto (I+{K}_{z}(\mathcal{M}+zI)^{-1})^{-1}$, and thus $\mathcal L_z^{-1}$,  are meromorphic in $\cc\setminus(\spec D_0\cup I_\Sigma)$ such that the only possible singularities at the point of $\specd(-\mathcal M)$ are simple poles with residues of finite rank. Combining this with the fact that the maps  $(D_0-z)^{-1}$, $\Phi_z\Lambda$, and $\Lambda \Phi_{\overline{z}}^\ast$ are holomorphic in $\rho(D_0)$, we conclude from the resolvent formula \eqref{Krein3} that $(\overline{D_{\epsilon,\tau}} -z)^{-1}$ has a meromorphic extension to  $\cc\setminus ( \spec D_0 \cup I_\Sigma)$ such that the coefficients in the Laurent series at the poles are finite-rank operators. Therefore, $\spece \overline{D_{\epsilon,\tau}}\subset \spec D_0 \cup I_\Sigma$ and concludes the proof of the theorem.
\end{proof}

The proof above gives more, namely, if $z\in \rho(D_0)\cap\rr$ is such that $z\in\spece  \overline{D_{\epsilon,\tau}}$ then $0\in\spece \mathcal L_z$. Combining this with Lemmas \ref{L relate D} and \ref{L BS}  we arrive at:
\begin{corollary} For any $z\in\rho(D_0)$ there holds
\begin{align*}
z\in \spece  \overline{D_{\epsilon,\tau}}  &\Leftrightarrow0\in \spece \mathcal L_z ,\\
z\in \specd  \overline{D_{\epsilon,\tau}}  &\Leftrightarrow0\in \specd \mathcal L_z.
\end{align*}
\end{corollary}

	\section*{Acknowledgments} 
	Part of this paper was written while B. Benhellal was working at Carl von Ossietzky Universit\"at Oldenburg and was supported by the Deutsche Forschungsgemeinschaft (German Research Foundation), project 491606144. V. Bruneau is partially supported by the ANR-24-CE40-2939-01 grant. B. Benhellal and  P. Miranda were partially supported by Fondecyt grant 1241983.


\end{document}